\documentclass[10pt, leqno,final]{article}

\usepackage[T1]{fontenc}		     
\usepackage[utf8]{inputenc}			 
\usepackage[english]{babel}
\usepackage{amsmath}
\usepackage{amssymb}
\usepackage{amsthm}
\usepackage{mathtools}		
\usepackage{mathrsfs}		
\usepackage{eucal}			

\usepackage{braket}			
\usepackage[all,pdf]{xy}		
\usepackage{mparhack}		
\usepackage{relsize}			

\usepackage{a4wide}		

\usepackage{booktabs}		
\usepackage{multirow}		
\usepackage{caption}		
\usepackage{rotating}
\usepackage{subfig}

\usepackage{varioref}		
\usepackage{footmisc}

\usepackage{graphicx}		
\usepackage{epsfig}

\usepackage{enumerate}		
\usepackage[colorlinks=true, linkcolor=black, citecolor=black,
urlcolor=blue]{hyperref}		
\mathtoolsset{showonlyrefs}
\input xy
\xyoption{all}
\usepackage{geometry}

\renewcommand{\=}{\coloneqq}			
\newcommand{\N}{\mathbb{N}}                     
\newcommand{\M}{\mathbf{M}} 

\newcommand{\LL}{\mathcal{L}}

\newcommand{\s}{\mathbf{S}}

\newcommand{\Con}{\mathcal{C}}

\newcommand{\Z}{\mathbb{Z}}                     
\newcommand{\R}{\mathbb{R}}                     
\newcommand{\C}{\mathbb{C}}                     
\renewcommand{\set}[2]{\left\{{#1}\mid{#2}\right\}}       
\newcommand{\dd}{\mathrm d}			
\newcommand{\Tr}{\mathrm{Tr\,}}               
\newcommand{\Id}{I}                   

\newcommand{\GL}{{\mathrm{GL}}}                 
\newcommand{\Lin}{{\mathcal{B}\,}}      
\newcommand{\Cl}{\mathcal{C}\,}      
\newcommand{\AAA}{\mathcal{A}^{ad}\,}      
\newcommand{\trace}{{\mathcal B_1}}      
\newcommand{\SSS}{\mathrm{Sym}}         
\newcommand{\Mat}{\mathrm{Mat}}         

\newcommand{\ispec}{\iota_{\scriptstyle{\mathrm{sp}}}}

\newcommand{\icon}{\iota_{\scriptstyle{\mathrm{deg}}}}
\newcommand{\spfl}{\mathrm{sf\,}}

\newcommand{\tr}{\mathrm{tr\,}}     
\newcommand{\Imm}{\mathrm{Im}}

\newcommand{\Real}{\mathrm{Re}}
\newcommand{\norm}[1]{\left\| #1 \right\|}			

\DeclareMathOperator{\sgn}{sgn}		
\DeclareMathOperator{\diag}{diag}		

\newcommand{\iMorse}[1]{m^-({#1})}		

\newcommand{\email}[1]{\href{mailto:#1}{\textsf{#1}}}
	
\numberwithin{equation}{section}
          
\newtheorem{thm}{\sc Theorem}[section]  
       
\newtheorem{lem}[thm]{\sc Lemma}            
\newtheorem{prop}[thm]{\sc  Proposition}     
        
\newtheorem{defn}[thm]{\sc Definition}     
\newtheorem{rem}[thm]{\sc Remark}      
\newtheorem{ex}[thm]{\sc Example}

\usepackage{bm}

\newcommand{\trasp}[1]{{#1}^\mathsf{T}}	
\renewcommand{\s}{\mathbb{S}}

\title{A generalized index theory for non-Hamiltonian system}
\author{Alessandro Portaluri\thanks{The author is partially supported by “Progetto di Ricerca GNAMPA - INdAM”, codice CUP\_E55F22000270001.}, Li Wu\thanks{The author is  supported by NSFC No.~$\mathrm{11425105}$.}, Ran Yang\thanks{The author is supported by NSFC No.~$\mathrm{12001098}$, the Doctoral research start-up fund of East China
University of Technology (No.~$\mathrm{DHBK2019204}$) and the State Scholarship Fund from the CSC.}}
\date{\today}

\begin{document}
	
 \maketitle
 
\begin{abstract}

Morse index theory  in the classical framework provides an equality between the
spectral properties of the second variation of the Lagrangian functional and the oscillation properties of the space of solutions of the associated boundary value problem. 

For carrying over a similar result in the non-Hamiltonian context which can be useful for  investigating  the dynamical properties of dissipative
systems,   we  introduce a new topological invariant, the so-called “degree-index”  defined in terms of the  Brouwer  degree of a suitable determinant map of a boundary matrix, which provides one possible substitute of the Maslov index in this non-Hamiltonian framework and finally  we prove the equality between the Morse index and the degree-index in this non-selfadjoint setting through a new abstract trace formula. 

In the last part of the paper we apply our theoretical results to some 1D reaction-diffusion systems.

%
%

\vskip0.2truecm
\noindent
\textbf{AMS Subject Classification:} 58J30, 47H11, 55M25, 58J20.
\vskip0.1truecm
\noindent
\textbf{Keywords:} Spectral flow, Morse index, Non-Hamiltonian systems, Reaction–diffusion equations.
\end{abstract}


\section*{Introduction}

The  relation between the spectral properties of a solution of a second order selfadjoint boundary value problem and the topological properties of the associated variational equation along such a solution induced by the phase flow and encoded by an intersection index in the Lagrangian Grassmannian context could be traced back to Morse.  

Starting by the early works of Jacobi,  Morse in his monumental book \cite{Mor32}  provided  a precise formula  between the  Morse index  of the second variation of the Lagrangian action functional along an extremal and the total number of focal points along the extremal, which in a modern language could be described in terms of the  Maslov index.   Since then,  many generalizations to higher order ordinary differential operators,  minimal surfaces and elliptic partial differential operators,  were proved during the last decades.

The common denominator for all of these variational problems is the symplectic structure. In his pioneering paper,  Duistermaat \cite{Dui76} first recognize the crucial role played by the  Maslov index, seen as an intersection index in the Lagrangian Grassmannian manifold, for describing the Morse index. 

However out of this situation, the situation is much more involved and, as far  as we know, very few results are known. (For further details we refer the interested reader to  \cite{BCCJM22} and references therein). 

Driven by the relevance and by the central role played in a number of both geometrical and dynamical questions, in this paper we establish a new index theory in a non-Hamiltonian context.  The main problems to face with are to define in this situation both sides of the equality appearing in the Morse index theorem, namely the  Maslov index and the functional analytic class of operators, including the most peculiar ones appearing in the applications, and admitting a Morse index. 

Inspired by the construction of the “conjugate index” provided by authors in  \cite{MPP05} for  study of the distribution of conjugate points along a  semi-Riemannian geodesic and later by \cite{PW20} in the more general  selfadjoint case, in this paper we define a new counting named {\em degree-index} which is defined by means of a suspension of the complexified family of second order non-selfadjoint  boundary value problems.  

In this new theoretical framework, the crossing instants appearing in the description of the Maslov index are replaced by  zeros  of the determinant of a “boundary matrix” defined in terms of the fundamental solution and of the boundary conditions. So, as one could easily imagine, this new topological invariant is defined in terms of the Brouwer degree of this determinant map whose vocation is to count, algebraically, the total number of zeros.  

The second problem to face with is to define the class of operators for constructing an index theory.  The key for finding the right functional analytic class is based on the fact that the Brouwer degree of the determinant map of the boundary matrix coincides with the trace of an operator-valued one-form. Such a trace, at least in the selfadjoint case, is equal to the Morse index (cfr. \cite{PW20} and references therein for further details) and this equality has been proved by using the localization property of another important topological invariant provided by the spectral flow of an induced path of unbounded selfadjoint Fredholm operators. 

The literature on this invariant in the selfadjoint case is quite broad. However, we are not aware of results dealing with the unbounded and  non-selfadjoint case. For linear and bounded operators on Banach spaces this invariant has been constructed by authors in \cite{ZL99}. (Cfr. also \cite{Gar10} and references therein). So, in this paper, as by-product, we provide a  possible generalization of this invariant that could be of an independent interest.

 Through this invariant  we were able  to handle both the problems with the Morse index and the Maslov index, providing a new and we believe interesting way to generalize the classical Morse index theorem through a spectral flow formula to the non-Hamiltonian setting.  
 
 Beside of the interests in providing such a non-selfadjoint index theory it is worth noticing that the definition of the degree-index is quite straightforward and doesn't require the construction of an intersection theory which is, in general, a pretty involved task both from a topological and differential viewpoint. Another important advantage in common with the selfadjoint case, is that this invariant is much easier to compute, at least numerically, with respect to the Maslov index.

In the last part of the paper, we apply our theoretical results for computing the Morse index of a steady state solution of a 1D reaction-diffusion systems on a bounded interval and under Dirichlet boundary conditions. In particular, we describe the degree-index for an interesting class of operators in the plane in terms of two families of conjugate points having a diophantine relation. Finally we construct a counterexample in this non-selfadjoint case to the equality  between the Morse index and the total number of conjugate points in a non-generic case. This final example which is in striking contrast with the classical selfadjoint case, in particular, put on evidence how far is the non-selfadjoint case by the classical one. 

The paper is organized as follows:

\tableofcontents


\section*{Acknowledgements}

We thank the anonymous reviewers whose comments and suggestions helped improve and clarify this manuscript.

\section{Trace formula and  crossings of the imaginary axis}
Let $(H,\langle \cdot,\cdot \rangle)$ be a complex separable Hilbert space with induced norm denoted by $|\cdot |$. We denote by $\Cl(H)$ the set of closed linear operators on $H$, by $\Lin(H)$ the $C^*$-algebra of the bounded and linear operators on $H$ whose norm will be denoted by $\|\cdot\|$, by $\Lin_0(H)$ be the ideal of compact  operators of $\Lin(H)$. 

Let $T \in \Lin_0(H)$ and  $T^*$ be its adjoint. It is  immediate to check  that  $T^*T$ is a compact self-adjoint non-negative definite linear operator. For $j \in \N^*$, we let  $ \lambda_j(T^* T)$  be the non-zero eigenvalues of $T^*T$ where each eigenvalue is repeated according to its own multiplicity and  we define the  {\em $j$-th singular value\/} or {\em $j$-th s number\/} as the non-negative number given by   $s_j(T)\=\big[\lambda_j(T^*T)^{1/2}\big]$.
We set $ \trace(H)\=\Big\{T\in \Lin_0(H)| \sum_{j=1}^\infty s_j(T)<+\infty\Big\}$. Endowing  $\trace(H)$ with the {\em trace norm\/} given  by  
\begin{equation}\label{eq:trace-norm}
\|T\|_\trace\= \sum_{j=1}^\infty s_j(T),
\end{equation}
then  $\trace(H)$  turns into a Banach space,  termed the space of {\em trace class operators\/} and its elements are   {\em trace class operators\/} on $H$. We refer the reader to \cite{PW20} and references therein for the basic properties of trace class operators.

 Given $A \in \Lin(H)$, we denote the spectrum and the resolvent set of $A$ by $\sigma(A)$ and $\rho(A)$, respectively. Recall that the resolvent of $A$ is defined by 
\[
R(A)\equiv R(\lambda, A):=(A-\lambda \Id)^{-1}, \qquad \lambda \in \rho(A).
\]
\begin{defn}
 $A \in \Cl(H)$ is termed a {\bf  trace class resolvent operator} if its resolvent operator $R(A)$ is in the trace class; i.e.  $R(A) \in \Lin_1(H)$.
\end{defn}
We start by observing that a uniformly continuous perturbation with linear and bounded operators of a trace class resolvent operator produces a family of  continuous trace class resolvent operator.    
\begin{lem} \label{lm:spectrum_trace_class}
Let  $\Omega\subset \C$ be a  compact subset and $A$ be a  selfadjoint  trace class resolvent operator. Given a continuous family $z \mapsto C_z$ in $ \Lin(H)$ parametrized by $\Omega$, we let $L_z\= A+ C_z$.  

Then we have:
\begin{enumerate}
	\item[(i)] For every $z \in \Omega$, $L_z$ is a trace class resolvent operator
	\item[(ii)] There exists $M>0$ such that  the map  
\[
\Omega \ni z\longmapsto R(L_z) \in \Lin_1(H)
\]
is continuous for every $\lambda\in\C$ with $ |\Im \lambda |>M$.
\end{enumerate}

\end{lem}
\begin{proof} We start proving item (i). 
	 By an algebraic manipulation, we get that 
	\begin{multline}
	\lambda \Id - L_z=(\lambda\Id - A)- C_z= (\lambda \Id -A)-C_z(\lambda \Id- A)^{-1}(\lambda \Id- A)\\ 
	=\Big[\Id- C_z(\lambda\Id-A)^{-1}\Big](\lambda \Id -A).
	\end{multline}
Since $\Omega$ is compact, then there exists $M>0$ such that $\|C_z\|< M $, for every $z\in \Omega$. We let $\lambda =a+i b$ with $a,b\in \R$ and we observe that, being $A$ selfadjoint, its spectrum is real. So, by a direct calculation, if $\lambda \in \sigma(A) \subset \R$, we get that  the spectrum of the operator $\Id- (A-a\Id)/(i b)$ is of the form   $1+ic$ where $c:=(\lambda-a)/b \in \R$. Then, we have $\norm{\left(\Id- \dfrac{A-a\Id }{i b}\right)^{-1}}\leq 1$.  Now, we let  $|b|>M$ and  we observe that 
	\[
	\norm{C_z(\lambda\Id -A)^{-1}}\le \dfrac{\norm{C_z}}{|b|}\norm{\left(\Id- \dfrac{A-a\Id }{i b}\right)^{-1}}\leq\dfrac{\norm{C_z}}{M}<1 \qquad \textrm{ for each }\, z \in \Omega.
	\]
	By this computation, we conclude that the operator $B_z:=\Id- C_z(\lambda\Id-A)^{-1}$ is invertible for each $z \in \Omega$. So, we get 
	\begin{multline}\label{eq:compact-resolvent-stuff}
		R(L_z)\equiv(\lambda\Id- L_z)^{-1}=(\lambda\Id -A-C_z)^{-1}=(\lambda \Id-A)^{-1}(\Id-C_z(\lambda \Id -A)^{-1})^{-1}\Longrightarrow \\ 
		R(L_z)=R(A)\, B_z^{-1}.
	\end{multline}
	Since $R(A)$ is in the trace class, $B_z$  is bounded and the set $\Lin_1(H)$ is an ideal we get that $R(L_z)$ is in the trace class for each $z \in \Omega$. This concludes the proof of the item (i). The second item is a straightforward consequence of the continuity of the operator path $z \mapsto C_z$.  
\end{proof} 

The next result provides a trace formula in terms of the eigenvalues crossing the imaginary line from left to right and viceversa. 

Before, we define the following rectangle $\Omega:=[0,1] \times [-M, \times M]\subset \C$ where $M$ is the same constant appearing in Lemma~\ref{lm:spectrum_trace_class}. \footnote{With abuse of notation, we identify the ordered pair $(t,s)$ with the complex number $z:=t+is$.}
\begin{prop}\label{thm:lemma-4}
Let  $R\in \Cl(H)$ be a closed and invertible operator having bounded inverse and such that $R^{-1}\in \Lin_1(H)$.  

For $t \in [0,1]$, we  let $ C_t \in \Lin(H)$ and we define the operators  $A_t:=R+C_t$  and  $A_z:=A_t+isI$ where $z:=t+is$. Moreover, we assume that the following two conditions hold: 
	\begin{itemize}
	\item the path $t\mapsto C_t$ is of regularity class $\mathscr C^1$ 
	\item there exists $c>0  $ such that $\sigma(A_z)\cap \{\pm c+is\ | \ s\in\R  \}=\emptyset$ for every $(t,s)\in \Omega$ and $A_z$ is invertible for all $(t,s)\in  \partial\Omega $	
	\end{itemize}
Then, we have
 
	\begin{equation}\label{eq:trace-formula-local}
		\dfrac{1}{2\pi i} \int_{\partial\Omega}\Tr \dd A_z A_z^{-1} = l-m
	\end{equation}
	where $l$ (resp. $m$) is  the number of eigenvalues (counted with multiplicity) crossing the imaginary axis from left to right (resp. from right to left). 
\end{prop}
 \begin{proof}
Let \[
	P_t \=-\dfrac{1}{2\pi i} \int_{\partial([-c,c]\times [-M\times M])}(A_t-\lambda I)^{-1}\dd\lambda .
	\]
	Since $A_t-\lambda I$ is invertible on $\partial([-c,c]\times [-M\times M])$, $P_t$ is well-defined and its rank is equal to the sum of the rank of  all spectral subspaces corresponding to the the eigenvalues of $A_t$ belonging to the region $[-c,c]\times [-M,M]$ of the complex plane.

In particular, we get that the path $t\to P_t$ is continuous and  of finite rank projectors; moreover  $P_tA_t=A_tP_t$ and   $\dim \Imm P_t$ is  constant.

Since we are interested in counting the eigenvalues crossing the imaginary line, as consequence of  the above discussion, without leading of generalities we can only restrict to consider the eigenvalues in the region $[-c,c]\times [-M,M]$ and so, we  we only need to  count the eigenvalues of $P_tA_tP_t$ crossing the imaginary axis.

The spectrum of  $P_tA_tP_t\big\vert_{\Imm P_t}$ is a subset of the spectrum of $A_t$ and it is given by a continuous unordered $n$-tuple of  complex numbers. By \cite[Theorem 5.2, pag. 109]{Kat80}, there exist $n$  single-valued, continuous function $\lambda_n$ such that the set
\[
\Lambda:=\Set{\lambda_k|1\le k\le n}
\]
corresponds to the set of the $n$ eigenvalues of  $t\mapsto P_tA_tP_t\big\vert_{\Imm P_t}$. By using this path of projectors, we get that the following decomposition holds 
\begin{multline}
A_z=M_zN_z, \textrm{ where } M_z=P_t+Q_t(A_t+is\Id)Q_t\\ \textrm{ and } N_z=Q_t+P_t(A_t+is \Id)P_t,\qquad \textrm{ for }  z=t+is \in \Omega
\end{multline}
By invoking Lemma~\ref{thm:importante}, we get 
\begin{equation}
	\dfrac{1}{2\pi i} \int_{\partial \Omega} \Tr \dd A_z A_z^{-1} = \dfrac{1}{2\pi i} \int_{\partial \Omega} \Tr \dd N_z N_z^{-1}+ \dfrac{1}{2\pi i} \int_{\partial \Omega} \Tr \dd M_z M_z^{-1}\\= \dfrac{1}{2\pi i} \int_{\partial \Omega} \Tr \dd N_z N_z^{-1}
\end{equation}
where the last equality follows being $ \Tr \dd M_z M_z^{-1}$ an exact one-form  and by Lemma~\ref{lm:trace_det}, we get  that 
\begin{equation}\label{eq:riduzione}
\dfrac{1}{2\pi i} \int_{\partial \Omega} \Tr \dd N_z N_z^{-1} =\dfrac{1}{2\pi i} \int_{\partial \Omega} \dd \log\det [P_tA_tP_t+is \Id]\Big\vert_{\Imm P_t} .
\end{equation}
We set $S_z\=[P_tA_tP_t+is \Id]|_{\Imm P_t}$ and we observe that the (RHS) in Equation~\eqref{eq:riduzione} coincide with the Brouwer degree $\deg(\det S_z,\Omega,0)$ of the map $\det S_z=\prod (\lambda_k(t)+i s)$ being, in fact, the  winding number of the function $\det S_z$. (Cfr. \cite[Pag. 4]{DM21}, for further details).   So, in order to conclude the proof we only need to prove that  
\[
\deg(\det S_z,\Omega,0)=l-m.
\]
So, let us consider the eigenvalue $\lambda_i$ and we observe that only four cases can occur (according to the sign of $\lambda_i$ at the ends). More precisely 
\begin{equation}\label{eq:cases}
	(I)\ \begin{cases}
		\Real \lambda_i(0)<0\\
		\Real \lambda_i(1)>0 
	\end{cases},\quad  
	(II)\ \begin{cases}
	\Real	\lambda_i(0)>0\\
	\Real	\lambda_i(1)<0
	\end{cases},\quad 
	(III)\ \begin{cases}
	\Real	\lambda_i(0)<0\\
	\Real	\lambda_i(1)<0
	\end{cases},\quad
	(IV)\ \begin{cases}
	\Real	\lambda_i(0)>0\\
	\Real	\lambda_i(1)>0
	\end{cases}.
\end{equation}
For each one of these four cases, we construct the following four homotopies  
\begin{multline}
	h_{(I)}(\sigma,t)\=(1-\sigma)\lambda_i(t) + \sigma\left(t-\dfrac12\right), \qquad h_{(II)}(\sigma,t)\=(1+\sigma)\lambda_i(t) + \sigma\left(-t+\dfrac12\right)\\
	h_{(III)}(\sigma,t)\=(1-\sigma)\lambda_i(t) -\sigma,\qquad h_{(IV)}(\sigma,t)\=(1-\sigma)\lambda_i(t) + \sigma,
\end{multline}
where $(\sigma,t)\in[0,1]\times[0,1]$.  Notice that for every $\sigma\in[0,1]$ all functions $h_{(I)}, h_{(II)}, h_{(III)}, h_{(IV)}$ have no zeros for $t=0$ and $t=1$.
Thus, there exist $h,k,l,m \in \N$ such that $\det S_z$ is a homotopic (through an admissible homotopy) to the following function
\[
\phi(t,s)\=(1+is)^k (-1+is)^h \left(t-\dfrac12+ is\right)^l\left(-t+\dfrac12+is\right)^m
\]
where $n=h+k+l+m$. Moreover the functions $s \mapsto 1+ is$ and  $s \mapsto -1+ is$ are homotopic to  the constant functions $1$ and $-1$ respectively through the homotopy $s \mapsto 1+ i\lambda s$ and $s \mapsto - 1+ i\lambda s$, for $\lambda \in [0,1]$ respectively. Thus, we have
\begin{align}
	\deg(\det S_z, \Omega, 0)&=
	\dfrac{1}{2\pi i} \int_{\partial \Omega} \dd \log\det S_z= \dfrac{1}{2\pi i}\int_{\partial\Omega} \dd \log\left[\left(t-\dfrac12+ is\right)^l\left(-t+\dfrac12+is\right)^m\right]\\&= \dfrac{1}{2\pi i}\int_{\partial \Omega}  \dd \log \left(t-\dfrac12+ is\right)^l+ \dfrac{1}{2\pi i}\int_{\partial \Omega}  \dd \log\left(-t+\dfrac12+is\right)^m\\&= l-m.
\end{align}
This concludes the proof. 
\end{proof}
\begin{rem}\label{rem:explain-of-spetral-flow}
	If $A_z$ is invertible for every $z\in\Omega$, then the (LHS) of Equation~\eqref{eq:trace-formula-local} vanishes since $\Tr \dd A_z A_z^{-1}$ is a closed operator-valued one form on $\Omega$ as follows by Lemma~\ref{thm:importante}. From the geometric viewpoint,  formula \eqref{eq:trace-formula-local} describes the net number of eigenvalues of $A_t$ crossing the imaginary line from the left to the right side as the parameter $t$  runs from 0 to 1. During this evolution, eigenvalues of $P_t$ belongs  to the rectangle $[-c,c]\times[-M,M]$. The role of $c$ is to guarantee the existence of decomposition in Lemma~\ref{thm:importante}.
\end{rem}

The  next technical result which is interesting in its own, states that  it is possible to construct an explicit linear perturbation in order the perturbed family only have isolated singular points. 
\begin{lem}\label{lm:analytic_path}
Let $A\in \Cl(H)$ be a selfadjoint operator having compact resolvent, 
let $A_t=A+C_t$ and we assume that the path $t \mapsto C_t$ is real analytic.  For $\lambda\in \R$, we let 
\[
 S_\lambda :=\Set{(t,s)\in[0,1]\times [-M,M]| \ A_t+i s\Id +\lambda \Id \textrm{ is not invertible}}.
 \]
Then, for each $c>0$, there exists $\lambda$ with  $|\lambda|<c$, the cardinality of $S_\lambda$ is finite.
\end{lem}
In what follows we refer to $S_\lambda$ as the {\bf singular set}.
\begin{proof}
		Since $[0,1]$ is compact,  it is enough to prove the result in a sufficiently small neighborhood of each $t\in [0,1]$. Without loss of generality, we only need to prove the result for a small neighborhood of $0$, namely, in  $S_\lambda^\epsilon=\{(t,s)\in S_\lambda\ | \ t\in[0,\epsilon] \}$  for some $\epsilon>0$. 
		
		Let $M>0$ such that,  for every $t \in [0,1]$,  $\|C_t\|< M$ and we let 
		\[
		\sigma(A+C_0)\cap i \R=\{\lambda_i, 1\le i\le k \}.
		\]
		Let $D_{\delta}=[-\delta,\delta]\times [-M,M]$. Chose $\delta >0$ sufficiently small such that $D_{\delta}\cap \sigma(A+C_0)=\{\lambda_i, 1\le i\le k\}$. Then, there exists $\epsilon>0$ such that $\sigma(A_t)\cap \partial D_\delta=\emptyset$ for $t\in [0,\epsilon]$.

Let $P_t= \displaystyle \dfrac{1}{2\pi i}\int_{\partial D_\delta} (z \Id -A_t)^{-1} dz$. Then, $P_t$ is the projection to the total eigenspace of $A_t$ corresponding to the eigenvalues of $A_t$ in $D_{\delta}$ and we have $P_tA_t=A_tP_t$.
Then 
\[
\sigma (P_tA_tP_t |_{\Imm P_t})=\sigma (A_t) \cap D_\delta.
\]
Since $t\mapsto A_t$ and $t\mapsto P_tA_tP_t$ are real analytic and $\dim \Imm P_t$ is finite, we can extend them to a complex analytic path  in  a neighborhood of $[0,\epsilon]$ that we will denote by  $P_z$ and $P_zA_zP_z$, respectively.

By \cite[Chapter II, Section 1]{Kat80}, there exist $k$ holomorphic  functions $\mu_1,\mu_2,\cdots,\mu_k$ representing the eigenvalues of $z\mapsto P_zA_zP_z|_{\Imm P_z}$ repeated according to their multiplicity, having only algebraic singularities and whose values  of all of their branches at $z$ coincide with the eigenvalues of $P_zA_zP_z|_{\Imm P_z}$  at $z$. So, we get 
\[
S_{\lambda}^\epsilon=\bigcup_{i=1}^k\Set{(t,s) \in [0,\epsilon]\times [-M, M]|\mu_i(t)+is +\lambda=0}.
\] 
Let $T_\lambda^\epsilon:= \Set{t \in [0,\epsilon]|(t,s)\in S_\lambda^\epsilon}$. We will show that $T_\lambda^\epsilon$ is a finite set for $\lambda\in \R$ a.e.

In fact, since the singularities of $\mu_i$ are discrete (being $\mu_i$ analytic functions), then for  $\lambda\in \R$ a.e., $T_\lambda^\epsilon$ doesn't contain any singularity of  $\mu_i$. On the contrary, we assume that $t_0$ is a singularity of $\mu_i$. By definition of $S_\lambda^\epsilon$, it follows that $is+\lambda=-\mu_i(t_0)$ and so $\lambda=-\Re {\mu_i(t_0)}$. Then, by choosing a $\lambda$ out of the images of all $\mu_i$, we conclude that $T_\lambda^\epsilon$ doesn't contain a singularity of any $\mu_i$. 

 Let $Q\subset \R$ be the set of all singularities of the $\mu_i$ for $i=1, \ldots k$.  For some  $t_0\notin Q$,  there exists  a neighborhood $V$ of $t_0$, such that  all branches of  the $\mu_i$  are single-valued functions on $V$.

Now, we only have to  show that for a single-valued holomorphic function $f$ (here $f$ is any of the branch of the $\mu_i$) which is defined on a compact neighborhood $V$ of some real number $t_0$, for  $\lambda\in \R$ a.e., there exist  only finite number of $t\in V\cap \R$ such that $f(t)\in i\R-\lambda$.

We assume that, for  $\lambda_0 \in \R$ there exists an infinite number of $t \in V \cap \R$ such that $f(t)+\lambda_0 \in i\R$. Setting $f(t)+\lambda_0=f_1(t)+if_2(t)$ for some  $f_1$ and $f_2$ real-valued and analytic functions, assuming the existence of infinitely many $t$   is equivalent to state that the equation $f_1(t)=0$ has an infinite many zeroes on $\R\cap V$. But, if so, by  the analytic continuation property  of $f_1$, we get  that $f_1(t)\equiv 0$ on  $\R\cap V$. So, we get $f(t)+\lambda_0=i f_2(t)$ for every $t \in V \cap \R$. Now, by choosing $\lambda\neq \lambda_0 $, we get that for every $t \in V\cap \R$ $f(t)+\lambda\notin i \R$. 

By this argument it follows that such a $\lambda_0$ should be isolated and in particular, for a.e. $\lambda \in \R$ we get only a finite number of $t \in V \cap \R$ such that $f(t)+\lambda\in i\R$. 

Since for $\lambda\in \R$ a.e., $T_\lambda^\epsilon$ is compact (being a closed subset of a compact set) and $T_\lambda^\epsilon\cap Q=\emptyset$ by the previous discussion, then there exists a finite covering by  compact sets $V_i\cap \R$ such that the branches of each $\mu_i$ are single-valued holomorphic functions on the $V_i$. Then, we conclude that, for $\lambda\in \R$ a.e., $T_\lambda^\epsilon$ is a  finite set for each $V_i$ and being in a finite number, we get that  $T_\lambda^\epsilon$ is a finite set for a.e. $\lambda \in \R$ and so $S_\lambda^\epsilon$. The conclusion follows by using the compactness of the unit interval.
\end{proof}

\begin{rem}
It is worth noticing  that Lemma \ref{lm:analytic_path}  still holds for a piecewise real-analytic path.
\end{rem}


\section{A new spectral flow formula}

As already observed the proof of a generalized Morse index theorem in the non
selfadjoint situation is through a sort of a spectral flow formula. One of the major step is to define the spectral flow in this more general situation and this is the focus of this section.


\subsection{The spectral flow for closed non-selfadjoint operators}

We start with the definition of the admissible class of operators that we will deal  with. 
\begin{defn}
Let $A \in \Cl(H)$ be a selfadjoint trace class resolvent operator. The operator $L \in \Cl(H)$ is termed {\bf admissible} if it is of the form 
\[
L=A+C
\]	
for some $C \in \Lin(H)$ and the set of all admissible operators on $H$ will be denoted by $\AAA(H)$.
\end{defn}
	
\subsubsection*{Analytic case} We start defining the generalized spectral flow for analytic admissible paths. So, let 
\[
L:[0,1] \to \AAA(H) \quad \textrm{ defined by} \quad L_t:=A+ C_t
\]
for $t \mapsto C_t$ being a real analytic path in $\Lin(H)$.  
\begin{defn}
We define the {\bf spectral flow} 	of the real analytic path $L$ as follows 
\begin{equation}\label{eq:spectral-flow-case-1}
	\spfl(L_t, t \in [0,1])=\dfrac{1}{2\pi i}\int_{\partial \Omega} \Tr d(L_z+\lambda I)(L_z+\lambda I)^{-1},
\end{equation}
where $\lambda$ is a sufficiently small positive real number and where 
\[
L_z= L_t+ is \Id \quad \textrm{ for } z=t+is \in \Omega \quad \textrm{ and }\quad \Omega:=[0,1]\times [-M, M]
\]
where $M$ has been defined  in Lemma~\ref{lm:spectrum_trace_class}.
\end{defn}
In the next result we prove that the spectral flow of an admissible analytic path is well-defined and is an integer.
\begin{lem}\label{lem:spectral-integer}
Under the above assumptions the  spectral flow of $L$ is well-defined. Moreover
\[
\spfl(L_t, t \in [0,1]) \in \Z.
\]
\end{lem}
\begin{proof}
We split the proof into two steps. \\
{\bf First step.}
{\em We prove the result under hyperbolic ends, namely 
\[
\sigma(L_0)\cap i\R=\emptyset \textrm{ and  } \sigma(L_1)\cap i\R =\emptyset.
\]
}
We start to observe that, being the hyperbolicity an open condition, there exists $\epsilon>0$ such that $\sigma(L_0+\lambda \Id)\cap i\R =\sigma(L_1+\lambda \Id)\cap i\R =\emptyset$ for every $\lambda \in (0,\epsilon)$. Then the spectral flow \eqref{eq:spectral-flow-case-1} is well defined. Moreover, the path of operator-valued one-form $\lambda \mapsto d(L_z+\lambda I)(L_z+\lambda I)^{-1}$ is $\mathscr C^1$ in $\Lin_1(H)$. Since the trace of such a   path  is bounded by  the trace norm, then the trace $\Tr d(L_z+\lambda I)(L_z+\lambda I)^{-1}$ is continuous with respect to $\lambda$. This fact, in particular, implies that the spectral flow defined by \eqref{eq:spectral-flow-case-1} is continuous with respect to $\lambda.$
	
Next we prove, by using Proposition~\ref{thm:lemma-4},  that this number is actually  an integer. Now, by invoking Lemma~\ref{lm:analytic_path}, we can choose $\lambda\in (0,\epsilon)$ such that the cardinality of the singular set $S_\lambda$ is finite and we let $S_\lambda=\{z_1,z_2,\cdots,z_n\}$ where $z_j=t_j+is_j$. We assume that the real part of  $k$ (with $k\leq n$) of such $z_j$ are  distinct and let us denote them by $t_1, t_2, \cdots, t_k$.   Then for every $j=1, \ldots, k$ we can choose sufficiently small $\epsilon_j>0$ to guarantee the existence of $c_j>0$ such that $\sigma(A_z)\cap\{\pm c_j+is\ | \ s\in\R   \}=\emptyset$ for every $z=t+is\in\Omega_j\=[t_j-\epsilon_j,t_j+\epsilon_j]\times [-M,M]$. By invoking Lemma~\ref{thm:importante} we know that $\Tr d(L_z+\lambda\Id)(L_z+\lambda \Id)^{-1}$ is a closed operator-valued one form on $\Omega\backslash(\underset{1\leq j\leq k}{\bigcup}\Omega_j)$. Moreover   $\Tr d(L_z+\lambda \Id)(L_z+\lambda \Id)^{-1}$ is well-defined on $\partial\Omega_j$. So, we get 
\[
\spfl(L_t,t\in[0,1])=\sum_{1\le j\le k} \dfrac{1}{2\pi i}\int_{\partial_{\Omega_j}} \Tr d(L_z+\lambda\Id)(L_z+\lambda \Id)^{-1}.
\] 
Since, by Proposition~\ref{thm:lemma-4} each of the  term $\displaystyle \frac{1}{2\pi i}\int_{\partial_{\Omega_j}} \Tr d(L_z+\lambda\Id)(L_z+\lambda \Id)^{-1}$ is an integer, we get that  $\spfl(L_t,t\in[0,1]) \in \Z$. \\

{\bf Second step.} We assume that the ends are possibly not-hyperbolic; i.e. $\sigma(L_0)\cap i\R$ and $\sigma(L_1)\cap i\R$ are possibly not-empty. 

If so, we can find $\epsilon>0$ such that  for $\eta\in (0,\epsilon)$
\[
\sigma(L_0+\eta\Id)\cap i\R=\emptyset \qquad \textrm{ and  }\qquad \sigma(L_1+\eta\Id)\cap i\R =\emptyset.
\]
Then, by using the first step $\spfl(L_t+\eta \Id, t \in [0,1])$ is continuous w.r.t.  $\eta\in (0,\epsilon)$ and it is integer-valued. In particular,  it is locally constant on $(0,\epsilon)$. So,  the  spectral flow as defined in Equation~ \eqref{eq:spectral-flow-case-1} is independent on $\lambda$. This concludes the proof.
\end{proof}

Before extending the definition of the spectral flow to  continuous paths, we  prove that the the spectral flows for two analytic paths with the same endpoints both coincide.
 \begin{lem}\label{lm:well-define-spec}
 	For $i=1,2$, we assume that the paths $t\mapsto C_i(t)$ parameterized on the unit interval  are real analytic and that   $C_1(0)=C_2(0)$ and $C_1(1)=C_2(1)$. Then, we have 
 	\begin{equation}
 	\spfl(A+C_1(t), t \in [0,1])=\spfl(A+C_2(t), t \in [0,1]).	
 	\end{equation}
 \end{lem}
 \begin{proof}
 	Let us consider the following  two-parameter family 
 	\[
 	L_{s,t}=A+ sC_1(t)+(1-s)C_2(t) \qquad \textrm{ for } s,t \in [0,1].
 	\]
 	Then, $\spfl(L_{s,t},t\in [0,1])$ is continuous w.r.t.$s$. Being an integer-valued function, it is locally constant.  This concludes the proof. 
 	\end{proof}
 	
 	\subsubsection*{Continuous case}
Now we are ready to define the spectral flow to continuous path.
\begin{defn}\label{def:spectral-flow-continuous}
Let $L:[0,1] \to \AAA(H)$ be an admissible continuous path. We define the {\sc spectral flow} of $L$  as 
\begin{equation}\label{eq:spectral-flow-continuous}
\spfl(L_t, t\in [0,1])\=\spfl(M_t, t \in [0,1]),
\end{equation}
where $M:[0,1] \to \AAA(H)$ is an analytic admissible path having the same endpoints of $L$.
\end{defn}
\begin{rem}
A constructive way for choosing an analytic path associated to $L$, is to consider the $n$-th Bernstein polynomial defined by: 
\[
L_n(t)=\sum_{1\le k\le n} \binom{n}{k} t^k(1-t)^{n-k}L\left(\dfrac{k}{n}\right).
\]
\end{rem}
\begin{lem}
Under above assumptions, the spectral flow $\spfl(A+D_t, t\in [0,1])$ in \eqref{eq:spectral-flow-continuous} is well-defined, namely, it is independent on the choice of $C_t$.
\end{lem}
\begin{proof}
The proof readily follows by invoking Lemma~\ref{lm:well-define-spec}.
\end{proof}


\subsection{Spectral flow for admissible paths and Morse index}

A classical and useful property of the spectral flow in the case of paths of bounded essentially positive Fredholm operators is that it gives a measure of the difference of the Morse indices at the ends. As we will prove in this section, this important property still holds in this more general situation. 

Let $T$ be a  compact resolvent hyperbolic operator, namely, its resolvent is compact and $\sigma(L)\cap i\R=\emptyset$. We assume that there exists $K>0$ such that 
\[
\sigma(T)\cap \set{z}{\Real z <-K}=\emptyset. 
\]
So, $T$ has only a finite number of  eigenvalues counted with multiplicity having negative real part. More precisely, by the spectral property of compact operators, there exists a  sufficiently small $\delta>0$ such that all eigenvalues of $T$ with negative real part are located in 
 $D\=[-K,-\delta]\times [-M,M]$. In particular, $\Imm P^-(T)$ is finite dimensional where 
 \begin{equation}\label{eq:projection-operator}
P^-(T)\=\dfrac{1}{2\pi i}\int_{\partial D} (z\Id-T)^{-1} dz.
\end{equation}
Like in the selfadjoint case, we introduce the following. 
\begin{defn}
	Let $T$ be  a compact resolvent hyperbolic operator such that  
	\[
	\sigma(T)\cap \set{z}{\Real z <-K}=\emptyset
	\]
	for some $K>0$. We define its {\bf Morse index} as follows:
	\begin{equation}
	m^-(T)\=\dim \Imm\, P^-(T)
	\end{equation}
	where $P^-(T)$ is defined in \eqref{eq:projection-operator}.
\end{defn}
As in the selfadjoint case, also  in  this case, the spectral flow of a path parameterized by a compact interval measures the change of the Morse index at the endpoints. 
\begin{thm}\label{thm:spectral-morse}
Let $L:[0,1] \to \AAA(H)$ be a continuous and admissible path and we assume that there is $K>0$ such that for every  $t\in [0,1]$, the operator  $L_t$ has no eigenvalues on the open half-plane  $\set{z}{\Real z <-K}$. 	 Then, we have
\[
\spfl(L_t,t\in [0,1])=\iMorse{L_0}-\iMorse{L_1}.
\]
\end{thm}
\begin{proof}
Without leading in generalities, it is enough  to prove the theorem  in the real analytic case. By using Lemma~\ref{lm:analytic_path} we assume that the set of crossing instants is finite; namely: 
\[
\Gamma:=\Set{t\in[0,1]|\sigma (L_t)\cap i\R \neq \emptyset}=\Set{t_1, \ldots, t_n}.
\]
By the definition of the  spectral flow  and by taking into account of the finite additivity of the integral, we decompose  the region into subregions like in the first step of proof of Lemma~\ref{lem:spectral-integer} and then we add all of them together. 
In this way, we reduce to prove the theorem in  the case in which we have only one crossing instant $t_0$.  By invoking Proposition~\ref{thm:lemma-4} once again, we get that 
\[
\spfl(L_t,t\in [0,1])= l-m
\]
where $l$ (resp. $m$) is the number of eigenvalues crossing from the left (resp. right) to the right (resp. left). So, by this description we get that 
\[
\iMorse{L_1}=\iMorse{L_0}-l+m.
\]
Putting these last two equalities together, the thesis readily follows.
\end{proof}


\subsection{A spectral flow formula}

In this section we will give a generalized Morse index theorem for second order differential operators with non-constant coefficients.

We denote by  $ \Mat(N, \R)$ be the set of $N\times N$ real matrices, by $\SSS(N,\R)$ the $N\times N$ real symmetric matrices. We let $P \in \mathscr C^1\big([0,1],\SSS(N)\big)$, $ Q\in \mathscr C^1\big([0,1], \Mat(N)\big)$, $ S\in \mathscr C^1\big([0,1], \SSS(N)\big)$ and we  assume that  $P(x)$ is non degenerate for each $x\in [0,1]$. These data defines  the  linear second order differential operator $\mathscr A$  as follows
\begin{equation}
	\mathscr A\=-\dfrac{d}{dx}\left[P(x) \dfrac{d}{dx}+ Q(x)\right]+ \trasp{Q}(x)\dfrac{d}{dx} +S(x),  \qquad x \in [0,1].
\end{equation}
Let now consider the continuous path $t\mapsto C_t$ in $\Mat(N)$ and for $z=t+is$ we let   $C_z=C_t+is \Id$ and we denote by $\mathscr C_t$ and $\mathscr C_z$ the corresponding multiplication operators and by $\mathscr A_t, \mathscr A_z$ the  associated differential operators. 
 
The corresponding differential equation  $\mathscr A_z u=0$ is  given by 
\begin{equation}\label{eq:morse_sturm}
-\dfrac{d}{dx}\big[P(x) u'(x) + Q(x) u(x)\big]+ \trasp{Q}(x) u'(x)+S(x) u(x)+C_z(x) u(x)=0, \qquad x \in [0,1].
\end{equation}
By setting   $v(x)\= P(x)u'(x)+ Q(x) u(x)$ and  $w(x)\equiv \trasp{\big(v(x), u(x)\big)}$,  Equation~\eqref{eq:morse_sturm} fits into the following linear first order system (in general not Hamiltonian as soon as $x\mapsto S(x)+C_z(x)$ is not a symmetric path)
\begin{equation}\label{eq:Hamiltonian-system}
	w'(x)= J B_z(x) w(x), \qquad x \in [0,1]
\end{equation}
where $\displaystyle J\= \begin{bmatrix}
	0 & -\Id\\ \Id & 0
\end{bmatrix}$
and
\begin{equation}
	B_z(x)\= \begin{bmatrix}
		P^{-1}(x) & - P^{-1}(x) Q(x)\\ - \trasp{Q}(x) P^{-1}(x) &  \trasp{Q}(x) P^{-1}(x) Q(x) - S(x)-C_z(x)
	\end{bmatrix}.
\end{equation}
\begin{defn}\label{def:bc}
	For $i=0,1$ we set  $ R_i\in \Mat(2N,\R)$. We define the  {\bf boundary operator $\mathcal R$} as follows
	\begin{equation}\label{eq:bc}
	\mathcal R(u)\=  R_0 \begin{bmatrix}
	P(0) u'(0)+ Q(0)\\
	u(0)
	\end{bmatrix}
	+ R_1 \begin{bmatrix}
	P(1) u'(1)+ Q(1)\\
	u(1)
	\end{bmatrix}
	\end{equation}
	where $'$ denotes the derivative with respect to $x$.
\end{defn}

\begin{ex}  Below  there are some choices of $R_0$ and $R_1$  corresponding to the Dirichlet, Neumann and periodic boundary conditions. 
	\begin{itemize}
		\item  Dirichlet case corresponds to choose
		\[
		R_0:=\begin{bmatrix}
		0 & I_n\\
		0 & 0
		\end{bmatrix}, \qquad  R_1:=\begin{bmatrix}
		0 & 0\\
		0 & I_n
		\end{bmatrix}
		\]
		\item Neumann  case corresponds to choose
		\[
		R_0:=\begin{bmatrix}
		I_n & 0\\
		0 & 0
		\end{bmatrix}, \qquad  R_1:=\begin{bmatrix}
		0 & 0\\
		I_n & 0
		\end{bmatrix}
		\]
		\item Periodic boundary conditions corresponds to choose
		\[
		R_0= \begin{bmatrix}
		\Id & 0 \\
		0 & \Id
		\end{bmatrix}, \qquad  R_1= - R_0.
		\]
	\end{itemize}
\end{ex}

Let 
\begin{equation}
	R_z\= R_0 +R_1 \psi_z(1), \qquad z=t+is \in \Omega
\end{equation}
where $\psi_z$ is the fundamental solution of the system given in Equation~\eqref{eq:Hamiltonian-system} and $\Omega:=[0,1]\times [-M,M]$ defined above. Then we introduce the {\bf determinant map\/} $\rho$ as follows:
	\[
	\rho: \Omega \ni z  \longmapsto \rho(z) :=  \det\,  R_z\in  \C.
	\]
For each $t \in [0,1]$, we denote by $\mathcal A$ (resp. $\mathcal A_t$) the operator $\mathscr A$  (resp. $\mathcal A_t$) acting on the domain 
\[
\displaystyle  \mathcal D\=\Set{u \in H^2([0,1], \R^N)| \mathcal R u=0}
\]
and we consider the complex extension operator $\mathcal A$ (resp. $\mathcal A_t$) by acting on $\mathcal C^\infty([0,1], \C^N)$. With a slight abuse of notation we will not distinguish between $\mathcal D$ and its complexification as well as between $\mathcal A$ (resp. $\mathcal A_t$) and its complex extension. 
Denoting by $\mathcal C_t$ and $\mathcal C_z$ the corresponding multiplication operators on $\mathcal D$ and we define  $\mathcal A_z : \mathcal D \subset L^2([0,1], \C^N) \to L^2([0,1], \C^N)$ to be the closed unbounded operator on $\mathcal D$ pointwise defined by
\begin{equation}\label{eq:Sturm-Liouville-equation}
\big( \mathcal A_z u\big)(x)\= \big(\mathcal A u  + \mathcal  C_zu)(x).
\end{equation}	
	
	\begin{defn}\label{def:incriminata}
		Under above notations, if $0\notin\rho(\partial\Omega)$ we define the {\bf degree-index} associated to the pair $(\psi, R)$ as the integer  $\icon(\psi, R)$  given  by
		\begin{equation}\label{eq:HPW-index}
		\icon (\psi, R) \= \deg ( \rho , \Omega , 0)
		\end{equation}
		and we define the {\em spectral index\/} of $\mathcal A$  as  the integer  $\ispec(\mathcal A)$  given by 
		\begin{equation}\label{eq:spectral-HPW-index}
		\ispec (\mathcal A) \= \spfl(\mathcal A_t, t \in [0,1]).
		\end{equation}
	\end{defn}
\begin{rem}
Definition~\ref{def:incriminata} is well-posed. This follows by  recalling that the according to  Definition~\ref{def:spectral-flow-continuous}	the spectral flow is well-defined for any continuous path in $\AAA(H)$. Since, it is a standard fact  that the operator $\mathcal A$ is selfadjoint trace class resolvent, we get that the path $t \mapsto\mathcal A_t$ is admissible. For further details, we refer the interested reader to  \cite[Lemma $3.8$]{PW20}. 
\end{rem}

The next result which is a generalization of the classical spectral flow formulas in the non-selfadjoint setting, provides an equality between the spectral index defined in terms of the spectral flow and the degree-index which encodes the properties of the solutions space of a differential operator and represents, in this framework the symplectic counterpart of the Maslov index. 
\begin{thm}\label{thm:spectral-degree}
	Under notation above and by assuming that $\ker \mathcal A_0= \ker \mathcal A_1 = \{0\}$, the following equality holds: 
	\begin{equation}
		\ispec(\mathcal A)= \icon(\psi, R).
	\end{equation}
\end{thm}
\begin{proof}
The proof of this result follows repetendum ad verbatim the proof of  \cite[Theorem 3.12]{PW20} provided in the for selfadjoint case. 
\end{proof}


\section{A non-selfadjoint Morse index theorem}

The aim of this section is to prove in this non-selfadjoint framework, the Morse index theorem under Dirichlet boundary conditions (which in the variational case correspond to the classical Bolza problem). 

In this section, if not otherwise stated we assume that the path  of symmetric matrices $x \mapsto P(x)$ (which pointwise represents the principal symbol of the differential operator) is positive definite. 

For each $l \in (0,1]$, we define the domain 
\[
  \mathcal D_{Dir}(l)\=\Set{u \in W^{2,2}([0,l], \R^N)|  u(0)=0=u(l)}= W^{2,2}([0,l], \R^N)\cap W_0^{1,2}([0,l], \R^N)
\]
Let $\mathcal A_z(l):\mathcal D(l)\subset L^2([0,l],\R^n)\to L^2([0,l],\R^n)$ be the closed unbounded operator pointwise having densely defined on $\mathcal D(l)$ and given by 
\[
	\mathcal A_z(l)u:=-\dfrac{d}{dx}\big[P(x) u' + Q(x) u\big]+ \trasp{Q}(x) u'+S(x) u+C_z(x) u,
\]
where $C_z(x)= C_0(x)+ t K\Id +is\Id $ with constant $K>0$ large enough and $z=t+is$. 
We assume that $x\mapsto C_0(x)$ is uniformly bounded by a constant $M$. 

Let $\psi_z(l)=\begin{bmatrix}E_z(l) & F_z(l)\\ G_z(l)& H_z(l)\end{bmatrix}$ be the fundamental solution of the associated first order system; then, for Dirichlet boundary condition we have
\begin{equation}
R_z=\begin{bmatrix}0& I_n \\0& 0\end{bmatrix} +\begin{bmatrix}
0&0\\0& I_n  
\end{bmatrix}\psi_z=\begin{bmatrix}
0& I_n\\ G_z& H_z
\end{bmatrix}.
\end{equation}
In particular the determinant map, reduces to  $\rho(z,l)=(-1)^n\det G_z(l)$.
\begin{thm}\label{thm:Morse_index_theorem}
	For $l \in (0,1]$, let  $\mathcal A_0(l):\mathcal D_{Dir}(l)\subset L^2([0,l],\R^n)\to L^2([0,l],\R^n)$ be the closed unbounded operator pointwise given by 
\[
	\mathcal A_0(l)u:=-\dfrac{d}{dx}\big[P(x) u' + Q(x) u\big]+ \trasp{Q}(x) u'+S(x) u+C_0(x) u. 
\]
	Then, we have 
	\begin{equation}
	\iMorse{\mathcal A_0(l)}=\deg ( \rho(0,l) , \Omega , 0)=\dfrac{1}{2\pi i}\int_{\partial V} \dd\log \det(G_{is}(x)),
	\end{equation}
	where $V=  [-M, M]\times [\delta, l]$ for some sufficiently small $\delta>0$.
\end{thm}
\begin{proof}
	We will give the proof by three steps.
	
{\bf First Step.} {\em We  prove that there exists $K>0$ large enough such that all eigenvalues of $\mathcal A_{1+is}(l)$ have positive  real part such that  $\iMorse{\mathcal A_{1+is}(l)}=0$ for every $s\in[-M,M]$ and $l\in(0,1]$.}

Let $u_l$ be an eigenvector of $\mathcal A_0(l)$ with eigenvalue $\lambda_l$.
We have
\[
(\mathcal A_0(l) \bar u_l,\bar u_l)+(\mathcal A_0(l) u_l,  u_l)=2\,\Real\,\lambda_l \|u_l\|_{L^2}^2 .
\] 
Let $\widetilde{\mathcal A}(l)$ be the self-adjoint operator in $L^2$ having domain $\mathcal D_{Dir}(l)$ and defined by 
\[
\widetilde{\mathcal A}u := -\dfrac{d}{dx}\big[P(x) u' + Q(x) u\big]+ \trasp{Q}(x) u'.
\]
We have 
\begin{equation}
\begin{aligned}
2\,\Real\,\lambda_l \|u_l\|_{L^2}^2&=(\widetilde{\mathcal  A}(l)  u,u)+(\widetilde{\mathcal  A}(l) \bar u, \bar u)+((\mathcal S+\mathcal C_0) \bar u,\bar u)+((\mathcal S+\mathcal C_0) u,u)\\
&\geq2 (\lambda_1(l)-\norm{\mathcal S+\mathcal C_0})\norm{u_l}^2_{L^2},
\end{aligned}
\end{equation}
where $\lambda_1(l)$ is the minimal eigenvalue of $\widetilde{\mathcal  A}(l)$ and where $\mathcal S, \mathcal C_0$ are the multiplication operator on $\mathcal D_{Dir}$ induced by $S, C_0$ respectively.

By the results in \cite[Section $6.2$]{CH53} we have $\lambda_1(l)$ tends to $+\infty$ as $l\to 0$. So $\Real\, \lambda_l$ is bounded from below uniformly on $(0,1]$. Moreover, there exists $\delta>0$ small enough such that all eigenvalues of $\mathcal A_0(x)$ have positive and arbitrarily large real part  for every $x\in (0,\delta)$. Consequently, we get $\iMorse{\mathcal A_0(\delta)}=0$. By choosing  $K>0$ large enough so that all eigenvalues of $\mathcal A_{1+is}(l)$ have real part with a uniform positive lower bound, we get that $\iMorse{\mathcal A_{1+is}(l)}=0$ for every $s\in[-M,M]$ and $l\in(0,1]$. 

{\bf Second Step.} {\em We prove that 
\[
\iMorse{\mathcal A_0(l)}=\dfrac{1}{2\pi i}\int_{\partial \Omega} \dd \log \det (G_z(l)).
\]
}

In fact, by Theorem \ref{thm:spectral-morse} we have 
\begin{equation}
\spfl( \mathcal A_t(l), t \in [0,1])=\iMorse{\mathcal A_0(l)}.
\end{equation}
Denote $\Omega=[0,1]\times[-M,M]$, then by Theorem \ref{thm:spectral-degree} there holds
\begin{equation}\label{eq:morse-index-A_0}
\iMorse{\mathcal A_0(l)}=\deg ( \rho_0(l) , \Omega , 0)=\deg(\det G_z(l), \Omega, 0)=\dfrac{1}{2\pi i}\int_{\partial \Omega} \dd \log \det (G_z(l))
\end{equation}
for every $l\in(0,1]$.

{\bf Third Step.} {\em We prove that 
\[
\dfrac{1}{2\pi i}\int_{\partial \Omega} \dd \log \det (G_z(l))=\dfrac{1}{2\pi i}\int_{\partial V} \dd\log \det(G_{is}(x))
\]
where $V= [\delta, l]\times [-M, M]$.}


By the first step, we know that there exists a sufficiently small positive real number $\delta$  such that $\iMorse{A_0(\delta)}=0$. Then for every $l>\delta$ by using Equation~\eqref{eq:morse-index-A_0}, we get  
\begin{equation}\label{eq:difference-morse-index}
\begin{aligned}
\iMorse{A_0(l)}&=\iMorse{A_0(l)}-\iMorse{A_0(\delta)}\\
&=\dfrac{1}{2\pi i}\int_{\partial \Omega} \dd \log \det(G_z(l))-\dfrac{1}{2\pi i}\int_{\partial \Omega}\dd\log\det(G_z(\delta)).
\end{aligned}
\end{equation}
The function  $\R^3\ni (t,s,x)\mapsto \det(G_{t+is}(x))\in \C$ is a  differentiable function on $\R^3$. So, the one-form $\dd \log \det(G_{t+is}(x))$ is a well-defined on the open set $\R^3\backslash \set{(z, x)}{\det(G_{t+is}(x))=0}$.
Let $W_i$ for $i=1, \ldots, 6$ the six faces of the parallelepiped defined by:
\begin{align*}
&	W_1=\set{(t,s,x)}{t=0,s\in [-M,M],x\in [\delta,l]},\\
&	W_2=\set{(t,s,x)}{t=1,s\in [-M,M],x\in [\delta,l]},\\
&	W_3=\set{(t,s,x)}{t\in [0,1], s=-M,x\in [\delta,l]},\\
&	W_4=\set{(t,s,x)}{t\in [0,1], s=M,x\in [\delta,l]},\\
&  W_5=\set{(t,s,x)}{t\in [0,1], s\in [-M,M], x=\delta},\\
&  W_6=\set{(t,s,x)}{t\in [0,1], s\in [-M,M],x=l}.\\
\end{align*}

By Equation \eqref{eq:difference-morse-index} and by taking into account the orientation, we get 
\begin{equation}\label{eq:morse-index-four-term}
\begin{aligned}
\iMorse{\mathcal A_0(l)}&=\dfrac{1}{2\pi i}\int_{\partial W_6} \dd \log \det(G_z(x))-\dfrac{1}{2\pi i}\int_{\partial W_5} \dd\log\det(G_z(x))\\
&=\dfrac{1}{2\pi i}\int_{\partial W_1} \dd\log\det(G_z(x))-\dfrac{1}{2\pi i}\int_{\partial W_2} \dd\log\det(G_z(x))\\
&+\dfrac{1}{2\pi i}\int_{\partial W_3} \dd\log\det(G_z(x))-\dfrac{1}{2\pi i}\int_{\partial W_4} \dd\log\det(G_z(x)).
\end{aligned}
\end{equation}
In order to conclude, we need to compute 
\[
\dfrac{1}{2\pi i}\int_{\partial W_i} \dd\log\det(G_z(x)) \quad \textrm{ for } i=1,2,3,4.
\]
By invoking the first step, the operator $\mathcal A_{1+is}(x)$ only has eigenvalues with positive real part; in particular,  $\mathcal A_{1+is}(x)$ is nondegenerate on $\partial W_2$. By \cite[Lemma 3.3]{PW20} $\det(G_{1+is}(x))$ is not zero in $\partial W_2$ and hence 
\[
\dfrac{1}{2\pi i}\int_{\partial W_2} \dd\log\det(G_z(x))=0.
\]
Since $\|C_0\|<M$, then the operator $\mathcal A_{t+is}(x)$ is nondegenerate for on $\partial W_3$ and on  $\partial W_4$ by Lemma~\ref{lm:spectrum_trace_class} and invoking once again \cite[Lemma 3.3]{PW20}, we get that   $\det G_{t+is}(x)$ is not zero in both $\partial W_3$ and $\partial W_4$; hence
\[
\dfrac{1}{2\pi i}\int_{\partial W_3} \dd\log\det(G_z(x))=\dfrac{1}{2\pi i}\int_{\partial W_4} \dd\log\det(G_z(x))=0.
\]
Summing up, we get that 
\[
\iMorse{\mathcal A_0(l)}=\dfrac{1}{2\pi i}\int_{\partial W_1} \dd\log\det(G_z(x)).
\]
The conclusion follows by observing that $\partial W_1$ coincides with the set $V:=[-M,M]\times [\delta, l]$  and $\det G_z(x)$ on $W_1$ reduces to $\det G_{is}(x)$. This concludes the proof. 
\end{proof}


\subsection{An explicit computation in the  constant case}

In this section we assume that  

\[
P=I_n, \qquad Q=0, \qquad  S+C_z =L+is\Id_n \quad \textrm{ for } \quad L\in \C^{n\times n}.
\]  
Then $JB_z=\begin{bmatrix}0 & L+is \Id_n \\ \Id_n & 0\end{bmatrix}$. So, the fundamental matrix solution is given by 
\begin{align}
\psi_z(x)&=\exp(x JB_z )=\sum_{k=0}^{+\infty}\dfrac{x^k}{k!} \begin{bmatrix}0 & L+is \Id_n \\ I_n & 0\end{bmatrix}^k  \\
&=\sum_{k=0}^{+\infty} \dfrac{x^{2k}}{(2k)!} \begin{bmatrix}(L+is\Id_n)^{k}&0\\ 0&(L+isI_n)^{k}\end{bmatrix}
+\dfrac{x^{2k+1}}{(2k+1)!}\begin{bmatrix}0 & (L+is \Id_n)^{k+1} \\  (L+is \Id_n)^{k} & 0\end{bmatrix}.
\end{align}
By this, we immediately get that the down-left block is given by 
\[
G_{is}(x)=\sum_{k=0}^{+\infty}  \dfrac{x^{2k+1}}{(2k+1)!}(L+is \Id_n)^{k}.
\]
If $L=\lambda \Id_n +N $ and  $N$ is a nilpotent  matrix and $\lambda\in \R$,  we have 
\[
G_{is}(x)=\sum_{k=0}^{+\infty}  \dfrac{x^{2k+1}}{(2k+1)!}(\lambda I_n +is I_n +N )^{k}=\sum_{k=0}^{+\infty}  \dfrac{x^{2k+1}}{(2k+1)!}(\lambda+is)^{k}I_n + N_2,
\]
where $N_2$ is a new nilpotent matrix obtained by collecting all the remaining terms of the binomial.  By this computation it follows that  $\det G_{is}(x)$ is  the same  after replacing  $\lambda \Id_n +N$ by $\lambda \Id_n$\footnote{This can be seen, for instance by using the Jordan normal form decomposition of $G_{is}(x)$.}. This implies that 
\[
\iMorse{-\dfrac{d^2}{dx^2}+\lambda \Id }=\iMorse{-\dfrac{d^2}{dx^2}+\lambda \Id +\mathcal  N }
\] 
where $\mathcal N$ is the multiplication operator induced by a nilpotent matrix.
By using Jordan decomposition, we can assume that if $L=L_s+L_n$ where $L_s$ is semi-simple and $L_n$ is nilpotent, then
\[
\iMorse{-\dfrac{d^2}{dx^2}+L}=\iMorse{-\dfrac{d^2}{dx^2}+L_s}.
\]
If  $P$  is a constant and positive definite matrix, by a  similar calculation, we get
\begin{equation}\label{eq:P-positive-matrix}
G_{is}(x)= \sum_{k=0}^{+\infty}  \dfrac{x^{2k+1}}{(2k+1)!}(P^{-1}L+is P^{-1})^k P^{-1}.
\end{equation}
Thus, we get 
\begin{equation}\label{eq:P-pisitive-deg}
\deg\det(G_{is}(x))=\deg\det \left(\sum_{k=0}^{+\infty}  \dfrac{x^{2k+1}}{(2k+1)!}(P^{-1}L+is P^{-1})^k\right)
\end{equation}


\subsubsection*{A striking difference w.r.t. Hamiltonian systems}

It is well-known fact, in the Hamiltonian framework, that if the Hamiltonian is $\mathscr C^2$ convex in the momenta, under Dirichlet boundary condition (otherwise this is in general not anymore true) the local contribution to the Maslov index or to the spectral flow, at each crossing instant (or verticality moment by using Arnol'd language) is positive. As direct consequence of the theory developed in \cite{PW20}, this implies also that the local contribution to the degree-index, in the classical framework is positive definite.

This property fails in the non-Hamiltonian case.  We will now construct an explicit simple non-selfadjoint operator in which the degree-index (in the Dirichlet bc) is negative. 
\begin{ex}\label{example-degree-minus-one}
The example is constructed by considering the following two matrices: 
\[
P=\begin{bmatrix}
1&0\\0&1/2
\end{bmatrix} \quad \textrm{  and } \quad P^{-1}L=\begin{bmatrix}
9/5&-4\\21/10&-4
\end{bmatrix}.
\]
By an explicit computation we get that the two eigenvalues of $P^{-1}L+isP^{-1}$ are
the given by: 
\begin{equation}
\lambda_1(s)=\frac{1}{10}(-11+15is-\sqrt{1-25s^2-290is}), \  \lambda_2(s)=\frac{1}{10}(-11+15is+\sqrt{1-25s^2-290is})
\end{equation}
and  $\lambda_2(0)=-1$. By an explicit computation and recalling the power series of the hyperbolic functions, we get that 
\begin{equation}
\det(G_{is}(x))=\lambda_1^{-1/2}\sinh(\lambda_1^{1/2}x)\lambda_2^{-1/2}\sinh(\lambda_2^{1/2}x).
\end{equation} 
By choosing $(s,x)=(0,\pi)$ and using the fact that $\lambda_2(0)=-1$, we get that 
\[
\det G_0(\pi)=0
\]
Given a  small neighborhood $\Omega$ of $(s,x)=(0,\pi)$,  we will show  that 
\[
\deg(\det(G_{is}(x)),\Omega,0)=-1.
\]
If $\Omega$ is sufficiently small and by taking the first order Taylor expansion, we get that, up to higher order terms:
\begin{equation}
\lambda_2(s)\sim_0\dfrac{1}{10}\Big[-11+15\,i\,s+1+\dfrac 1 2 (-290\,i\,s)\Big]=\dfrac{1}{10}(-10-130\,i\,s)=-1-13\,i\,s.
\end{equation}
In particular
\begin{align}
(-\lambda_2(s))^{1/2} \sim_0 (1+13\,i\,s)^{1/2}\sim_0 1+6.5\,i\,s.
\end{align}
Finally, we get 
\begin{align}
\sinh\big((\lambda_2(s))^{1/2}x\big)&=i\sin\big((-\lambda_2(s))^{1/2}x\big)=i\sin\big((-\lambda_2(s))^{1/2}x-\pi+\pi\big)\\&=-i\sin\big((-\lambda_2(s))^{1/2}x-\pi\big)\sim_0-i((1+6.5\,i\,s)x-\pi)\\&=6.5\,s\,x-i\,x+\,i\pi.
\end{align}
Being $\lambda_1^{-1/2}\sinh(\lambda_1^{1/2}x)\lambda_2^{-1/2}\neq 0$ in a sufficiently small neighborhood of $(0,\pi)$, then we get
\begin{align}
\deg(\det(G_{is}(x)),\Omega,0)=\deg(6.5\,s\,x-ix+i\pi,\Omega,0)=-1
\end{align}
where the last equality direct follows by looking at the sign of the determinant  of the Jacobian matrix induced by the map
\[
F: \R^2 \to \R^2 \quad \textrm{ defined by } \quad F(s,x)=\begin{bmatrix}
	6.5\,s\,x \, \pi\\ -x+\pi
\end{bmatrix}
\]
corresponding to the real and imaginary part of the complex function $6.5\,s\,x-ix+i\pi$.
\end{ex}


\section{Application to reaction–diffusion systems}

This section is devoted to an explicit computation of the Morse index for a non-selfadjoint boundary value problem arising by considering the linearization  of  a reaction-diffusion system about a steady state solution. Furthermore, we will also study the relationship between the eigenvalues having negative real parts of the corresponding differential operator  and the number of conjugate points along the solution. 

Given a positive $L$, we start by considering the  1D reaction–diffusion system on $(0,+\infty) \times (0,L)$ given by
\begin{equation}\label{eq:rec-dif-sys}
u_t=Du_{xx}+F(u), 
\end{equation} 
where $u(t,x)\in \R^n,$ $D=\diag(d_1,\cdots,d_n)$ is the diagonal matrix having positive entries  $d_i>0$ for each $i=1, \ldots, n$ and $F \in \mathscr C^\infty(\R^n,\R^n)$. We recall that a {\bf steady state solution $u^*$} with Dirichlet boundary condition  is a solution of the  following bvp:
\[
\begin{cases}
Du_{xx}+F(u)=0\qquad \textrm{ on } (0,L) \\
u(0)=0= u(L)
\end{cases}
\]
After linearizing the system given in  in Equation~\eqref{eq:rec-dif-sys}, we end-up  
with the linear parabolic system given by: 
\begin{equation}\label{eq:ori-linear-sys}
v_t=D\frac{d^2}{dx^2}v+\nabla F(u^*)v.
\end{equation}
We let $\mathcal L$  be denoting the  linearized operator on $L^2((0,L), \R^n)$ having  dense domain $\mathcal D(\mathcal L):=W^{2,2}((0,L), \R^n) \cap W_0^{1,2}((0,L), \R^n) $ given by:
\begin{equation}
\mathcal L\=D\frac{d^2}{dx^2}+\nabla F(u^*): \mathcal D(\mathcal L) \subset L^2((0,L), \R^n)\longrightarrow L^2((0,L), \R^n).
\end{equation}
If not otherwise stated, from now on we only consider the case $n=2$ and we assume that 
\[
D=\begin{bmatrix}
1&0\\0&d
\end{bmatrix}, \qquad \qquad \textrm{ where } d>0.
\]
For the sake of convenience,  we set  $V:=\nabla F(u^*)=\begin{bmatrix}
v_{11}&v_{12}\\v_{21}&v_{22}
\end{bmatrix}$. 
\begin{rem}
If $F=\nabla G$ for some smooth function $G$, then $V$ is symmetric. In particular the differential operator $\mathcal L$ is selfadjoint.
\end{rem}
From now on,  we assume following two conditions hold
\begin{equation}\label{eq:stable-no-diff}
\tr V<0;\quad \det V>0
\end{equation}
and \begin{equation}\label{eq:unstable-with-diff}
v_{22}+dv_{11}>2\sqrt{d\det V}
\end{equation}
hold. In fact, condition \eqref{eq:unstable-with-diff} is equivalent to 
\begin{equation}\label{eq:unstable-with-diff-equi}
\M^2-4d\det V>0\quad \text{and}\quad \M>0.
\end{equation}
 where $\M=v_{22}+dv_{11}$. Therefore $d\neq 1$ since $\tr V<0$.
\begin{rem}
We observe that the above two conditions are classical in the stability analysis. . More precisely, they come out  from the Turing instability problem of the homogeneous steady state in a reaction-diffusion system. (Cfr. the interested reader to \cite[Section 2.3]{Mu03}) Roughly speaking, it means that the steady state is stable under small perturbations in the absence of diffusion but unstable under small spatial perturbations when diffusion is present. Turing is the first to study this phenomenon in \cite{Tur52}. In \cite[Section 2.3]{Mu03}, the author discuss the above two conditions for reaction-diffusion systems with zero flux (Neumann) boundary condition. It is worth mention that, recently  the authors in \cite[Section $5.3$]{BCCJM22} studied the Turing instability by using a Maslov-type index theory constructed for  non-Hamiltonian systems.
\end{rem}


\subsection{Computation of the degree-index for 1D reaction-diffusion  systems}

The aim of this section is to perform an explicit computation of the degree-index for steady state solutions of 1D reaction-diffusion systems under the two  conditions given in Equation~\eqref{eq:stable-no-diff} and  Equation~\eqref{eq:unstable-with-diff}. 

Let 
\begin{equation}
-\mathcal L(s)\=-D\dfrac{d^2}{dx^2}-V+isI
\end{equation}
Under the notation above,  and setting  $P=D$ and $L=-V$,  by a direct calculation we get 
\[
P^{-1}L+isP^{-1}=\begin{bmatrix}
-v_{11}+is&-v_{12}\\-\dfrac 1 d v_{21}&-\dfrac 1 d v_{22}+\dfrac 1 dis
\end{bmatrix}.
\] 
In particular, it holds that 
\[
\tr \left(P^{-1}L+isP^{-1}\right)=\dfrac{1}{d}\big[-\M +is(d+1)\big], \qquad \det \left(P^{-1}L+isP^{-1}\right)= \dfrac{1}{d}\big[\det V -is\,\tr V-s^2\big].
\]

Letting
\begin{equation}
G(s,x)=\sum_{k=0}^{+\infty}\dfrac{x^{2k+1}}{(2k+1)!}\left(P^{-1}L+isP^{-1}\right)^k,
\end{equation}
then, it follows that $G(s,x)=G_{is}(x)P$ and so
\[
\deg\det(G_{is}(x))=\deg\det(G(s,x)). 
\]
Denoting by  $\lambda_{\pm}(s)$  the eigenvalues of matrix $P^{-1}L+isP^{-1}$, then we get 
\begin{equation}
\lambda_{\pm}(s)=\dfrac{-\M+(d+1)is\pm\sqrt{(\M^2-4d\det V)-(d-1)^2s^2+[4d\tr V-2(d+1)\M]is}}{2d}.
\end{equation}
\paragraph{ First case: $s \neq 0$.} In this case, if $\lambda_+(s)=0$ (or $\lambda_-(s)=0$), then we get that 
\begin{equation}
\det V=s^2\quad \text{and}\quad 
4\,d\ \tr V=0.
\end{equation}
We observe that these two conditions are incompatible with $d>0$ and $\tr V< 0$.  By this discussion, we get that   $\lambda_+(s)$ and  $\lambda_-(s)$ are both different from zero. Therefore, 
\begin{equation}
\sum_{k=0}^{+\infty}\dfrac{x^{2k+1}}{(2k+1)!}\lambda_{\pm}^k(s)=\dfrac{1}{\sqrt{\lambda_{\pm}(s)}}\sinh(\sqrt{\lambda_{\pm}(s)}x)
\end{equation}
and 
\begin{equation}
\det G(s,x)=\frac{1}{\sqrt{\lambda_{+}(s)}}\sinh(\sqrt{\lambda_{+}(s)}x)\cdot \frac{1}{\sqrt{\lambda_{-}(s)}}\sinh(\sqrt{\lambda_{-}(s)}x).
\end{equation}
$\sinh(\sqrt{\lambda_{+}(s)}x)=0 $ iff  $\sqrt{\lambda_{+}(s)}x=k\pi i$ for some $k\in \N$ and so $\lambda_{+}(s)=-(k^2\pi^2)/x^2$. Now, we show that this is not the case. In fact, we get
\begin{multline}
 -\M+(d+1)is+\sqrt{(\M^2-4d\det V)-(d-1)^2s^2+[4d\, \tr V-2(d+1)\M]is}=-\dfrac{2dk^2\pi^2}{x^2}\\
(\M^2-4d\det V)-(d-1)^2s^2+[4d\,\tr V-2(d+1)\M]is=\left[-\dfrac{2dk^2\pi^2}{x^2}+\M-(d+1)is\right]^2.
\end{multline}
By equating the imaginary parts of both sides, we get that 
\[ 
\tr V =\dfrac{(d+1)k^2\pi^2}{x^2}>0
\]
which is incompatible with $\tr V <0$. Therefore, we have  $\sinh(\sqrt{\lambda_{+}(s)}x)\neq 0 $   for every $s\neq 0$ and $x>0$. The same arguments hold for  $\lambda_{-}(s)$.

\paragraph{\bf Second case: $s=0$.} Then 
\begin{equation}
\lambda_{\pm}(0)=\frac{-\M\pm\sqrt{\M^2-4d\det V}}{2d}.
\end{equation}
Moreover, under the conditions given in Equation~\eqref{eq:stable-no-diff} and Equation~\eqref{eq:unstable-with-diff-equi}, it is easy to check that  both $\lambda_-(0)$ and $ \lambda_+(0)$ are real and negative. 

We assume that $(s,x)=(0,x_0)$ is a solution of  $\sinh(\sqrt{\lambda_{+}(0)}x)=0$, then $\sqrt{\lambda_{+}(0)}x_0=k\pi i$ for some $k\in\N$. Now, we  compute the degree-index $\deg\det(G(s,x),\Omega,0)$.  Choosing a sufficient small neighborhood  $\Omega$  of $(0,x_0)$ such that $\sinh(\sqrt{\lambda_{-}(s)}x)\neq 0$ for every $(s,x)\in \Omega$, then we have 
\begin{align}\label{eq:degree-simplify}
\deg(\det G(s,x),\Omega,0)&=\deg\Big(\frac{1}{\sqrt{\lambda_{+}(s)}}\sinh(\sqrt{\lambda_{+}(s)}x)\cdot \frac{1}{\sqrt{\lambda_{-}(s)}}\sinh(\sqrt{\lambda_{-}(s)}x),\Omega,0\Big)\\
&=\deg(\sinh(\sqrt{\lambda_{+}(s)}x),\Omega,0   ) +\deg(\sinh(\sqrt{\lambda_{-}(s)}x),\Omega,0   ) .
\end{align}

We need a technical lemma to calculate the degree.
	\begin{lem}\label{lm:cal_degree}
		We assume that  $f(s)=a+ibs +o(s)$ with $a, b\in \R\setminus{\{0\}}$ and that $\sinh (\sqrt{f(0)}x_0)=0$. Then
		we get 
		\[
		\deg(\sinh(\sqrt{f(s)}x),\Omega, 0  )=-\sgn(ab)
		\]
\end{lem}
\begin{proof}
	Since   $ g: z\to \sqrt z $ is analytic on small neighborhood of some $z\neq 0$ and $g'(z)\neq 0$, for every $z_0 \neq 0$, then we have
	$\deg (\sqrt z, \Omega_1,  z_0)=1$  where $\Omega_1$ is some small neighborhood of $z_0^2$ (cfr. \cite[Corollary 6.4.1, pag. 316]{DM21}).
	Similarly, we have $\deg(\sinh z, \Omega_2,  0)=1$ where $\Omega_2$ is some small neighborhood of $k\pi i$ with $k\in \N$.
	Then, by using the Leray product formula (cfr. \cite[Theorem 7.4.3]{DM21}), we can conclude that 
	\[
	\deg(\sinh(\sqrt{f(s)}x),\Omega, 0  )=\deg(f(s)\,x^2,\Omega, 0  ),
	\]
	where $\Omega$ is some small neighborhood of  $(0,x_0)$ with $ f(0)x_0^2=-k^2\pi^2$.
	Let $f(s)=u(s)+ i v(s)$. Then, we have
	\[
	\left | \begin{matrix} u'(s)x^2 & v'(s)x^2\\2x u(s)& 2x v(s)\end{matrix} \right |=2x^3(u'v-v'u).
	\]
Since $f(s)= a +i\,bs + o(s)$, it holds that $u(s)=a, v(s)=bs.$ Then it follows that
\[
\deg(\sinh(\sqrt{f(s)}x),\Omega, 0  )=\deg(f(s)x^2,\Omega, 0  )=\sgn(u'(0)v(0)-v'(0)u(0))=-\sgn(ab).
\]
This concludes the proof. 
\end{proof}

Now, we need to estimate $\lambda_{+}(s)$. For simplicity, we denote
	\begin{equation}
		\Delta_1\=\M^2-4d\det V>0,\quad \Delta_2\=4d\tr V-2(d+1)\M<0,
	\end{equation}
	then 
	\begin{align}
		\lambda_{+}(s)&=\frac{-\M+(d+1)is+\sqrt{\Delta_1-(d-1)^2s^2+\Delta_2is}}{2d}\\
		&=\frac{-\M+(d+1)is+\sqrt{\Delta_1}\sqrt{1+(-(d-1)^2s^2+\Delta_2is)/\Delta_1}}{2d}\\
		&\sim \frac{-\M+(d+1)is+\sqrt{\Delta_1}(1+\Delta_2is/(2\Delta_1))}{2d}\\
		&=\frac{-\M+\sqrt{\Delta_1}}{2d}+\frac{2(d+1)\sqrt{\Delta_1}+\Delta_2}{4d\sqrt{\Delta_1}}is.
	\end{align}
Note that  $-\M +\sqrt {\Delta_1} <0$ and $2(d+1)\sqrt{\Delta_1}+\Delta_2=2(d+1)(\sqrt{\Delta_1}-\M)+4d\tr V<0$ being the sum of two negative quantities.
By Lemma~\ref{lm:cal_degree}, we have 
\begin{equation}\label{eq:degree-1}
\deg(\sinh(\sqrt{\lambda_{+}(s)}x,\Omega,0)=-1,
\end{equation}
where $\Omega$ is a small neighborhood of $(0,x_0)$ such that $\lambda_+(0)x_0^2=-k^2\pi^2$. By the same arguments and by taking into account the minus sign in front of the square root, we have 
\begin{equation}\label{eq:degree+1}
\deg(\sinh(\sqrt{\lambda_{-}(s)}x,\Omega,0)=1.
\end{equation}
In order to compute the contribution to the degree at each zero, for the sake of convenience, we introduce the following notation. 

Given a proposition $\mathcal{P}$, we set 
\[
[\mathcal{P}]:= \begin{cases}1 \quad \textrm{ if $\mathcal{P}$ is true}
\\0 \quad \textrm{ if $\mathcal{P}$ is false }
\end{cases}.
\]
Then by Equation~\eqref{eq:degree-simplify},  we get  
	\begin{multline}\label{eq:local_degree}
	\deg(\det(G(s,x)),\Omega,0)=\left[\sqrt{-\lambda_-(0)x_0^2/\pi^2}\in \N \right]-\left[\sqrt{-\lambda_+(0)x_0^2/\pi^2}\in \N \right]\\
	=\begin{cases}
	1\quad \ \ \textrm{ if $\sqrt{-\lambda_-(0)x_0^2/\pi^2}\in \N$ and $\sqrt{-\lambda_+(0)x_0^2/\pi^2}\notin \N$}\\[8pt]
	-1\quad \textrm{ if $\sqrt{-\lambda_-(0)x_0^2/\pi^2}\notin \N$ and $\sqrt{-\lambda_+(0)x_0^2/\pi^2}\in \N$}\\[8pt]
	0\quad\ \ \textrm{ if $\sqrt{-\lambda_-(0)x_0^2/\pi^2}\in \N$ and $\sqrt{-\lambda_+(0)x_0^2/\pi^2}\in \N$}\\
	\end{cases}
	\end{multline}
where $\Omega$ is a sufficiently small neighborhood of $(0,x_0)$ such that $(0,x_0)$ is the only zero in $\Omega$.

\subsection{Negative real eigenvalues and conjugate points}

Let $a>0$ and we assume that conditions \eqref{eq:stable-no-diff} and \eqref{eq:unstable-with-diff} hold.   We consider the eigenvalue problem 
\begin{equation}
-\LL v=-D\frac{d^2}{dx^2}v-Vv=\lambda v,\quad x\in[0,a]
\end{equation}
under  Dirichlet boundary conditions for $\lambda \in (-\infty,0)$. 
Here we stress the fact that, from now on, we will  assume that the operator$-\LL$ is nondegenarate.

By using Equation~\eqref{eq:P-positive-matrix}, we can define 
\begin{equation}
G(\lambda,x)=\sum_{k=0}^{+\infty}\dfrac{x^{2k+1}}{(2k+1)!}(-P^{-1}V-\lambda P^{-1})^k,
\end{equation}
where $-P^{-1}V-\lambda P^{-1}=\begin{bmatrix}
-v_{11}-\lambda&-v_{12}\\-\dfrac 1 dv_{21}&-\dfrac 1 d(v_{22}+\lambda)
\end{bmatrix}$. Then we have 
\begin{equation}
\lambda\in\sigma(-\LL)\quad\iff\quad  \det G(\lambda,a)=0.
\end{equation}
We  note that
	\[
	\det G(\lambda,a)=\dfrac{1}{\sqrt{\mu_1}}\sinh(\sqrt{\mu_1}a)\frac{1}{\sqrt{\mu_2}}\sinh(\sqrt{\mu_2}a),
	\]
	where $\{\mu_1,\mu_2\}=\sigma(-P^{-1}V-\lambda P^{-1})$.
By a direct calculations we get
\begin{equation}\label{eq:eigenvalue-matrix-condition}
\mu\in \{\mu_1,\mu_2\}\ \iff \ \mu^2+\frac 1 d (\M+(d+1)\lambda)\mu+\frac 1 d (\lambda^2+\lambda \tr V+\det V)=0.
\end{equation}
Note that $\mu_1,\mu_2 \neq0$, otherwise there holds $\lambda^2+\tr V\lambda+\det V=0$ which is impossible for $\lambda<0$ being $\tr V <0$ and $\det V>0$.

So, we get 
$\lambda\in \sigma (-\LL)$ if and only if $ \det G(\lambda,a)=0$ which imples that  the equation 
\[
 \mu^2+\dfrac 1 d (\M+(d+1)\lambda)\mu+\dfrac 1 d (\lambda^2+\tr V\lambda+\det V)=0 
 \]
 has a solution $\mu$ such that $\sinh(\sqrt{\mu}a)=0$.

So, $
\lambda\in \sigma (-\LL)\cap \R^-$ if and only if  $\lambda$ is a negative solution of 

\begin{equation}\label{eq:eigen-matrix-condition-equi}
	\lambda^2+((d+1)\mu+\tr V)\lambda+d\mu^2+\M\mu+\det V=0,
\end{equation}
with $\mu =-(k^2\pi^2)/a^2$ for some $k\in \N^+$.

Let $\lambda_i, i=1,2$ be the solutions of equation \eqref{eq:eigen-matrix-condition-equi}. Then, we get three cases
\begin{itemize}
\item[(1)]  $\lambda_1<0$ and $\lambda_2<0$
\item[(2)]  $\lambda_1<0$ and $\lambda_2=0$
\item[(3)]  $\lambda_1<0$ and $\lambda_2>0$.
\end{itemize} 
Since $\lambda_1+\lambda_2=-((d+1)\mu+\tr V)>0$, then cases (1) and (2) cannot happen.

In the  case (3) we have 
\begin{equation}\label{eq:condition-for-eigen}
\lambda_1\cdot\lambda_2=d\mu^2+\M\mu+\det V=:g(\mu)<0.
\end{equation}
Recalling that $\Delta_1= \M^2 -4\,d \det V$ we denote by $\mu_-$ and $\mu_+$ the two distinct real and negative zeros of $g(\mu)=0$. By this discussion we get that 
%
the number of negative eigenvalues of operator $-\LL$ are the following:
	\begin{align}\label{eq:number-of-negative-eigen}
		\#\Set{\lambda\in\sigma(-\LL)\cap\R^-}=\#\Set{k\in\N^+\ \Big|\dfrac{\M-\sqrt{\Delta_1}}{2d}a^2<k^2\pi^2<\dfrac{\M+\sqrt{\Delta_1}}{2d}a^2}.
	\end{align}

\begin{rem}
The eigenvalues in the  (LHS) of Equation~\eqref{eq:number-of-negative-eigen}  are counted according to their multiplicities.
It is possible that even if two different $k_1,k_2$  define two different $\mu_1$ and $\mu_2$, the solution $\lambda$ of the Equation~\eqref{eq:eigen-matrix-condition-equi} is the same. In this case   $\lambda$ is a negative eigenvalue of $-\LL$ with multiplicity $2$.
\end{rem}
We can prove that  operator $-\LL$ has no eigenvalues such that $\lambda\in\C \setminus \R $ and  $\Real \lambda<0$. We observe that 
\[
\iMorse{-\LL} \ge \# \Set{\textrm{negative eigenvalues of } -\LL}.
\]
We will show that  the equation holds and then all the eigenvalues of $-\LL$ with negative real parts are all negative eigenvalues.

By invoking Theorem~\ref{thm:Morse_index_theorem}, we have 
\[
\iMorse{-\LL}= \deg(\det(G_{is}(x)),\Omega,0),
\]
with $\Omega=[-M,M]\times [\delta, a]$.
We observe that  if $\det G_{is}(x)=0$ then we have $s=0$. 

So, by Equation~\eqref{eq:local_degree} we have
\begin{multline}
\iMorse{-\LL}=\deg(\det(G_{is}(x)),\Omega,0)=\sum_{0<x_0<a} \left[\sqrt{-\lambda_-(0)x_0^2/\pi^2}\in \N \right]-\left[\sqrt{-\lambda_+(0)x_0^2/\pi^2}\in \N \right]\\
=\sum_{ 0<k\pi <\sqrt{-\lambda_-(0)}a}1 -\sum_{ 0<k\pi<\sqrt{-\lambda_+(0)}a} 1=\#\Set{k\in\N^+\ \Big| \ \dfrac{\M-\sqrt{\Delta_1}}{2d}a^2<k^2\pi^2<\dfrac{\M+\sqrt{\Delta_1}}{2d}a^2}\\
\overset{\eqref{eq:number-of-negative-eigen}}{=}  \#\Set{\textrm{ negative real eigenvalues  of } -\LL}.
\end{multline}
We note that since, by assumption, the operator $-\LL$ is nondegenerate,  
\[
((\M-\sqrt{\Delta_1})a^2)/(2d)\neq k^2\pi^2.
\]
By this equality we get that the Morse index $\iMorse{-\LL}$ of $\LL$,  counting all eigenvalues having negative real part coincides with the set of all negative real eigenvalues. In particular, there exists no complex not real eigenvalue having negative real part. 

Next, we compute the number of conjugate points in $(0,a)$. We recall that $x_0\in (0,a)$ is called a conjugate point if $-\LL v=0$ has a nontrival solution $v$ such that $v(0)=v(x_0)=0.$ In fact, we have 
	\begin{equation}
		x_0\ \text{is a conjugate point in} \ (0,a)\ \iff\ G(0,x_0)=\sum_{k=0}^{+\infty}\dfrac{x_0^{2k+1}}{(2k+1)!}(-P^{-1}V)^k \ \textrm{ is degenerate.}
	\end{equation}
	If $\mu\in\sigma(-P^{-1}V)$, then $\mu\neq0$ since $\det V>0.$ Being $\lambda=0$, we observe that in order  $\lambda=0$ to be  a solution of  Equation~\eqref{eq:eigen-matrix-condition-equi}, then  $\mu$ has to be a zero of $g(\mu)$. By the previous characterization of $\det G(\lambda, a)$, we get that since $\mu\in\sigma(-P^{-1}V)$ then	 $\dfrac{1}{\sqrt{\mu}}\sinh(\sqrt{\mu}x_0)\in\sigma(G(0,x_0))$. Therefore, if 
	\begin{multline}\label{eq:conjugate-condition-equi}
		G(0,x_0)\ \textrm{ is degenerate }\ \Rightarrow \sinh(\sqrt{\mu}x_0)=0.  \\
		\textrm{ Then } \sqrt{\mu}x_0=k\pi i\ \iff \mu=-\dfrac{k^2\pi^2}{x_0^2}<0
	\end{multline}
	for some $k\in\N^+$.  If $\mu=\mu_-,$ then we have
	\begin{equation}
		\dfrac{-\M-\sqrt{\Delta_1}}{2d}=-\dfrac{k^2\pi^2}{x_{01}^2} \quad \Rightarrow \quad x_{01}=\sqrt{\frac{2d}{\M+\sqrt{\Delta_1}}}\cdot k\pi.
	\end{equation}
	If $\mu=\mu_+,$ then we have
	\begin{equation}
		\dfrac{-\M+\sqrt{\Delta_1}}{2d}=-\dfrac{k^2\pi^2}{x_{02}^2}\quad \Rightarrow\quad  x_{02}=\sqrt{\dfrac{2d}{\M-\sqrt{\Delta_1}}}\cdot k\pi.
	\end{equation}
	Clearly, there holds $0<x_{01}<x_{02}$. Then, without considering the multiplicity we have  
	\begin{align}
		&\#\{\text{Conjugate points in } (0,a)\ \text{without multiplicity}\}\\&=\#\Big\{ x_0\in (0,a)\ \Big| \ \exists\ k\in\N^+ \ \text{such that } x_0= \sqrt{\frac{2d}{\M+\sqrt{\Delta_1}}}\cdot k\pi \ \text{or } x_0= \sqrt{\dfrac{2d}{\M-\sqrt{\Delta_1}}}\cdot k\pi \Big\}\\
		&=\#\Big\{ x_0\in (0,a)\ \Big| \ \exists\ k\in\N^+ \ \text{such that } x_0= \sqrt{\dfrac{2d}{\M+\sqrt{\Delta_1}}}\cdot k\pi  \Big\}+\\
		&\#\Big\{ x_0\in (0,a)\ \Big| \ \exists\ k\in\N^+ \ \text{such that } x_0= \sqrt{\dfrac{2d}{\M-\sqrt{\Delta_1}}}\cdot k\pi  \Big\}-\\
		&\#\Big\{ x_0\in (0,a)\ \Big| \ \exists\ k_1, k_2\in\N^+ \ \text{such that } x_0= \sqrt{\dfrac{2d}{\M+\sqrt{\Delta_1}}}\cdot k_1\pi=\sqrt{\dfrac{2d}{\M-\sqrt{\Delta_1}}}\cdot k_2\pi  \Big\}.
	\end{align}
	
	Now, if $x_0$ is a conjugate point in $(0,a)$ having multiplicity $2$ then this is equivalent to the fact that $0$ is an eigenvalue of $G(0,x_0)$ with multiplicity $2$. Therefore, we have 
	\begin{equation}
		\frac{1}{\sqrt{\mu_1}}\sinh(\sqrt{\mu_1}x_0)=\frac{1}{\sqrt{\mu_2}}\sinh(\sqrt{\mu_2}x_0)=0
	\end{equation}
	or equivalent to 
	\begin{equation}\label{eq:non-generic}
		x_0^2=-\dfrac{k_1^2\pi^2}{\mu_1}=-\frac{k_2^2\pi^2}{\mu_2}\iff\dfrac{\mu_1}{\mu_2}=\dfrac{k_1^2}{k_2^2},\quad \text{where } k_i\in\N^+ \ \textrm{ and } k_i<\dfrac{a\sqrt{-\mu}}{\pi}, i=1,2..
	\end{equation}
	We define following three sets of conjugate points:
	\begin{align}\label{eq:conjugate-points-set}
		\Con_1&=\Big\{ x_0\in (0,a)\ \Big| \ \exists\ k\in\N^+ \ \text{such that } x_0= \sqrt{\frac{2d}{\M+\sqrt{\Delta_1}}}\cdot k\pi  \Big\}\\
		\Con_2&=\Big\{ x_0\in (0,a)\ \Big| \ \exists\ k\in\N^+ \ \text{such that } x_0= \sqrt{\frac{2d}{\M-\sqrt{\Delta_1}}}\cdot k\pi  \Big\}\\
		\Con_3&=\Big\{ x_0\in (0,a)\ \Big| \ \exists\ k_1, k_2\in\N^+ \ \text{such that } x_0= \sqrt{\frac{2d}{\M+\sqrt{\Delta_1}}}\cdot k_1\pi\\&=\sqrt{\frac{2d}{\M-\sqrt{\Delta_1}}}\cdot k_2\pi  \Big\}.
	\end{align}
	In particular, each conjugate point in $\Con_3$ has multiplicity $2$. Then, by taking into account  the multiplicity we get 
	\begin{align}
		\#\{\text{Conjugate points in } (0,a) \ \text{with multiplicity }\}=\#\Con_1+\#\Con_2.
	\end{align}
	We define following three sets 
	\begin{multline}\label{eq:number-conjugate-equi}
		\Con_1'=\Set{ k\in \N^+\ \Big| \ 0< \sqrt{\dfrac{2d}{\M+\sqrt{\Delta_1}}}\cdot k\pi <a}=\Set{k\in \N^+\ \Big| \ 0<  k^2\pi^2 <\dfrac{\M+\sqrt{\Delta_1}}{2d}a^2}\\
		\Con_2'=\Set{ k\in \N^+\ \Big| \ 0< \sqrt{\dfrac{2d}{\M-\sqrt{\Delta_1}}}\cdot k\pi <a}=\Set{k\in \N^+\ \Big| \ 0<  k^2\pi^2 <\dfrac{\M-\sqrt{\Delta_1}}{2d}a^2}\\
		\Con_3'=\left\{k_1\in\N^+\ \Big| \ \exists\ k_2\in\N^+ \text{such that }\right.\\ \left. 0<  \sqrt{\dfrac{2d}{\M+\sqrt{\Delta_1}}}\cdot k_1\pi=\sqrt{\dfrac{2d}{\M-\sqrt{\Delta_1}}}\cdot k_2\pi <a\right\}.
	\end{multline}

	Being
	\[
	\sqrt{\dfrac{2d}{\M+\sqrt{\Delta_1}}}<\sqrt{\dfrac{2d}{\M-\sqrt{\Delta_1}}},
	\]
	then, we get  $\Con_3'\subset\Con_2'\subset \Con_1'.$ Moreover, we have $\#\Con_i=\#\Con_i', i=1,2,3.$
	
	For every $x_0\in \Con_i, i=1,2,3$, we assume $\Omega $ is a sufficiently small neighborhood of $(0,x_0)$ such that $(0,x_0)$ is the only one singularity. By formula \eqref{eq:local_degree}, we have 
	\begin{equation}
		\deg(\det G(s,x), \Omega, 0)=\begin{cases}
			1 \quad \textrm{ if }  x_0\in\Con_1\ \text{but}\ x_0\notin\Con_2 \\
			-1 \quad \textrm{ if }  x_0\notin\Con_1\ \text{but}\ x_0\in\Con_2 \\
			0 \quad \textrm{ if }  x_0\in\Con_1\cap\Con_2 \\
		\end{cases}
\end{equation}

Therefore, combining Equation~\eqref{eq:number-conjugate-equi} with Equation~\eqref{eq:number-of-negative-eigen}, we get 
\begin{multline}\label{eq:index-equal-eigen-number}
\icon (\psi, R)=\#\Con_1- \#\Con_2\\
=\Big\{ k\in \N^+\ \Big| \ \dfrac{\M-\sqrt{\Delta_1}}{2d}a^2<  k^2\pi^2 <\dfrac{\M+\sqrt{\Delta_1}}{2d}a^2 \Big\};\\
=\#\{ \lambda\in\sigma(-\LL)\ | \ \lambda<0  \}.
\end{multline}
We note, once again, that since we assume that the operator $-\LL$ is nondegenerate, so $((\M-\sqrt{\Delta_1})a^2)/(2d)\neq k^2\pi^2$. 
\begin{rem}
	 Condition \eqref{eq:non-generic} is analogous to what authors in  \cite[Definition $5.2$]{BCCJM22} termed  \textbf{non-generic} in the special case in which the entry $d$ of the matrix $D$ is equal $1$. In fact, under condition \eqref{eq:non-generic} we have $\#\Con_3\neq0$ which is equivalent to the fact that there exists a conjugate point in $(0,a)$ having  multiplicity $2.$ In this case, the equality 
	 \begin{equation}\label{eq:eigen-equal-conju}
	\#\Set{\lambda\in\sigma(-\LL)\ | \ \lambda<0}=\#\Set{\textrm{Conjugate points in } (0,a)\ \textrm{ without multiplicity}}
	 \end{equation}
may fail. In the quoted paper authors provide a proof of the previous equality under some generic assumptions. 

We observe that if $\M\gg 4d\det V$, then we have 
\[
\dfrac{\M-\sqrt{\Delta_1}}{2d}a^2\sim 0
\]
and consequently $\Con_2=\emptyset$. In particular Equation~\eqref{eq:eigen-equal-conju} holds. Moreover, in this case the index 
\[
\icon (\psi, R)\geq 0.
\]
\end{rem}


\section{Some counterexamples}

In this section we construct a couple of non-generic examples of linear planar second order non-selfadjoint systems in which the equality between the the total number of negative and real eigenvalues and the total number of conjugate points, fails.

\begin{ex}\label{example-non-generic}
We let $D=\Id_2$, we set $V=\begin{bmatrix}
4&1\\0&9 \end{bmatrix}$ and we choose $a=4$.
The associated bvp is the following
\begin{equation}
\begin{cases}
u''(x)+Vu(x)=0\qquad x\in [0,4]\\
u(0)=0=u(4)
\end{cases}
\end{equation} 
It is easy to check that condition provides at Equation~\eqref{eq:non-generic} is satisfied, namely, it is \textbf{ non-generic}. We will show that Equality \eqref{eq:eigen-equal-conju} does not hold.

Note that in this case $a=4, d=1, \M=\tr V=13, \det V=36.$  Then Equation~\eqref{eq:eigen-matrix-condition-equi} fits into the following one
\begin{equation}\label{eq:example-eigen-condition}
\lambda^2+(2\mu+13)\lambda+\mu^2+13\mu+36=0.
\end{equation}
Let $\mu=-\dfrac{k^2\pi^2}{16}<0$ for some $k\in\N^+.$ Since $(2\mu+13)^2-4(\mu^2+13\mu+36)=25>0$, then Equation \eqref{eq:example-eigen-condition} has to have two different real solutions $\lambda_1, \lambda_2$ and we assume that  $\lambda_1<\lambda_2$. By a direct computation we have $g(\mu)\=\mu^2+13\mu+36$ has two different zeros $\mu_1=-9, \mu_2=-4$.  About the sign of the solutions $\lambda$ of Equation~\eqref{eq:example-eigen-condition}, only three cases can occur. 

\paragraph{Case One:}  $\lambda_1<\lambda_2<0$.

This is equivalent to the following two facts
\begin{align}
&\lambda_1+\lambda_2=-(2\mu+13)<0 \quad \textrm{ and } 
\lambda_1\cdot\lambda_2=\mu^2+13\mu+36>0;\\
&\iff \mu>-\dfrac{13}{2};\quad \mu>\mu_2=-4\ \iff \mu=-\dfrac{k^2\pi^2}{16}>-4.
\end{align}
Then we have $0<k<\dfrac{8}{\pi}$ and consequently $k=1$ or $k=2$. In this case $-\LL$ has $4$ negative eigenvalues. 

\paragraph{Case Two:}  $\lambda_1<\lambda_2=0$.

This is equivalent to 
\begin{align}
&\lambda_1+\lambda_2=-(2\mu+13)<0  \quad \textrm{ and }\quad \lambda_1\cdot\lambda_2=\mu^2+13\mu+36=0\\
&\iff \mu>-\dfrac{13}{2}\quad \textrm{ or } \quad \mu=\mu_2=-4\ \iff \mu=-\dfrac{k^2\pi^2}{16}=-4
\end{align}
and this last equality is never satisfied.   In this case $-\LL$ has no negative eigenvalues.

\paragraph{Case Three:}  $\lambda_1<0<\lambda_2$.

This is equivalent to 
\begin{align}
 \lambda_1\cdot\lambda_2=\mu^2+13\mu+36<0 
\iff   -9=\mu_1<\mu=-\dfrac{k^2\pi^2}{16}<\mu_2=-4.
\end{align}
In this case,  we get that the above inequality holds iff $k=3$. In this case $-\LL$ has $1$ negative eigenvalue. 

Summarizing, the operator $-\LL$ has $5$ negative eigenvalues each one of multiplicity $1$.  In particular, the operator $-\LL$ is nondegenerate. If not,  
there exists a vanishing eigenvalue, namely  $0=\lambda_1<\lambda_2$. However, this case can be ruled out precisely as in the \textbf{case two}. Thus we have 
\begin{equation}
\#\{ \lambda\in \sigma(-\LL)\ | \ \lambda<0  \}=\#\{ \lambda\in \sigma(-\LL)\ | \ \lambda\leq 0  \}=5.
\end{equation}

Next we compute the conjugate points. Since $\sigma(-V)=\{-4, -9\}$, then by Equation~\eqref{eq:conjugate-condition-equi}, we have
\begin{equation}
x_0\ \text{is a conjugate poin in } (0,4)\iff x_0=\frac{k\pi}{2} \text{or }\frac{k\pi}{3} \text{for some } k\in\N^+ \text{and } x_0\in(0,4).
\end{equation}
By this discussion, we get that the  conjugate points are $\pi/2$ with multiplicity $1$, $\pi/3$ with multiplicity $1$, $2\pi/3$ with multiplicity $1$ and $\pi$ with multiplicity $2$ and moreover, we can check that $x_0=4$ is not a conjugate point. Therefore,  we have 
\begin{multline}
\#\{  \text{ Conjugate points in } (0,4)\ \text{without multiplicity} \}\\=\#\{  \text{Conjugate points in } (0,4]\ \text{without multiplicity} \}=4 \\ 
\#\{ \text{Conjugate points in } (0,4)\ \text{with multiplicity} \}\\=\#\{  \text{Conjugate points in } (0,4]\ \text{with multiplicity} \}=5.
\end{multline}
Then 
\begin{equation}
\#\{ \lambda\in \sigma(-\LL)\ | \ \lambda<0  \} \neq\ \#\{  \text{ Conjugate points in } (0,4)\ \text{without multiplicity} \}.
\end{equation}	
\end{ex}
\smallskip

We observe that in Example~\ref{example-non-generic},  both conditions \eqref{eq:stable-no-diff} and \eqref{eq:unstable-with-diff} are not satisfied. 

We now construct an example in which both conditions  \eqref{eq:stable-no-diff} and \eqref{eq:unstable-with-diff} hold but nevertheless the equality provided at Equation~\ref{eq:eigen-equal-conju} fails.

\begin{ex}
	Let $D=\begin{bmatrix}
	1&0\\0&1/2
	\end{bmatrix}$, $V=\begin{bmatrix}
	-1&-2\\49/128&3/4
	\end{bmatrix}$ and we choose $a=16.$

In this case, we have 
\begin{equation}
d=1/2,\ \tr V=-1/4<0,\ \det V=1/64>0,\ \M=1/4>2\sqrt{d\det V}=\sqrt{2}/8.
\end{equation}
Hence  both conditions given at Equation~\eqref{eq:stable-no-diff} and Equation~\eqref{eq:unstable-with-diff} hold. Then
\begin{equation}
\Delta_1=\M^2-4d\det V=1/32>0.
\end{equation} 
 Consequently,
 \begin{equation}
 \lambda_{+}(0)=\dfrac{-\M+\sqrt{\Delta_1}}{2d}=\dfrac{\sqrt{2}-2}{8}<0\qquad  \lambda_{-}(0)=\dfrac{-\M-\sqrt{\Delta_1}}{2d}=-\dfrac{\sqrt{2}+2}{8}<0.
 \end{equation} 

We now compute the conjugate points. By Equation~\eqref{eq:conjugate-points-set}  we have
\begin{align}
\Con_1&=\Set{ x_0\in (0,16)\ \Big| \ \exists\ k\in\N^+ \ \text{such that } x_0= \sqrt{\dfrac{8}{\sqrt{2}+2}}\cdot k\pi}\\&=\Set{\sqrt{\dfrac{8}{\sqrt{2}+2}}\cdot\pi,\ \sqrt{\frac{8}{\sqrt{2}+2}}\cdot 2\pi, \sqrt{\dfrac{8}{\sqrt{2}+2}}\cdot 3\pi}\\
\Con_2&=\Set{x_0\in (0,16)| \exists\ k\in\N^+ \ \text{such that } x_0= \sqrt{\dfrac{8}{2-\sqrt{2}}}\cdot k\pi}=\Set{\sqrt{\dfrac{8}{2-\sqrt{2}}}\cdot \pi}\\
\Con_3&=\Set{ x_0\in (0,16)| \exists\ k_1, k_2\in\N^+ \ \text{such that } x_0= \sqrt{\dfrac{8}{\sqrt{2}+2}}\cdot k_1\pi=\sqrt{\dfrac{8}{2-\sqrt{2}}}\cdot k_2\pi }=\emptyset.
\end{align}
Therefore, we have $\#\Set{\textrm{Conjugate points in } (0,16)}=4.$

Next, we compute the negative eigenvalues of $-\LL.$ By using Equation~\eqref{eq:number-of-negative-eigen} we have
\begin{align}
\#\Set{\lambda\in\sigma(-\LL)| \lambda<0 }&=\#\Set{ k\in\N^+\ \Big| \ \dfrac{2-\sqrt{2}}{8}\cdot 256<k^2\pi^2<\dfrac{2+\sqrt{2}}{8}\cdot 256 }=2.
\end{align}
Then by Equation~\eqref{eq:index-equal-eigen-number} we have
\begin{equation}
\icon (\psi, R)=^{\#}\{ \lambda\in\sigma(-\LL)\ | \ \lambda<0  \}=2.
\end{equation}
So, also in this case the total number of negative eigenvalues which is  $2$ differs to the total number of conjugate points which is $4$.
\end{ex}

\newpage

\appendix

\section{Operator-valued one form: a finite dimensional reduction}

In this section we provide the basic results mainly used in the proof of Proposition~\ref{thm:lemma-4}. For the proof, we refer the interested reader to \cite[Section 2]{PW20}.

We start by recalling that the spectrum of an operator  having compact resolvent consists of isolated points of finite multiplicity (cf. \cite[Theorem 6.29, pag. 187]{Kat80}). We assume   $A_t$ has compact resolvent;  then,   for each $t\in [0,1]$, there  exist positive numbers $c $ and $ M$ (depending on $t$) such that $\sigma(A_{t}) \cap   \s_{c,M}=\emptyset$, where $\sigma(A_{t})$ denotes the spectrum of $A_{t}$ and  $ \s_{c,M}\=\{\lambda=t+is \in \C\ | \ |t|= c,|s|\leq M\ \text{or}\ |t|\leq c, |s|=M\}$.

For each $t\in [0,1]$, we set   
\[
P_t \=-\dfrac{1}{2\pi i} \int_{\s_{c,M}}(A_t-\lambda I)^{-1}\dd\lambda 
\]
the projection onto the total eigenspace corresponding to the
eigenvalues of $A_t$ inside the rectangle $\{t+is\in\C\ | \ |t|\leq c, |s|\leq M\}$ and we observe that   $P_tA_t=A_tP_t$.
We set  $Q_{t}\= \Id - P_{t}$ and $z=t+is$ we define the following two operators
\begin{equation}\label{eq:M-N}
N_{z}\=  Q_{t}+ P_{t}\, (A_t+isI)\, P_t \quad \textrm{ and } \quad  M_z= P_t+ Q_t\,(A_t+isI)\,  Q_t.
\end{equation}
Setting $A_z\=A_t+isI$ and by a direct calculation, for each $z\in \Omega:=[0,1] \times [-M, M]$, the following holds
\begin{equation}\label{eq:nuova}
A_z=M_zN_z=Q_tA_zQ_t+ P_tA_zP_t.
\end{equation}

\begin{lem}\label{thm:importante}
	Let $X \subset \C$ be open,  $A_{z_0}$ be an invertible operator such that $A_{z_0}^{-1}\in  \trace(H)$ and for each $z \in X$, we let 
	\[
	A_z\=A_{z_0}+C_z.
	\]
	where $C_z \in \Lin(H)$. We assume that the map $z\mapsto C_z$  is of class $\mathscr C^1$ and we assume that  for some $c>0$,  	 $\sigma(A_{z_0})\cap \s_{c,M}=\emptyset$.
	 
Then, there exists a convex neighborhood $W_{z_0}$  of $z_0$ such that   
\[
A_z=M_zN_z \qquad  \textrm{ for every }z \in W_{z_0}
\]
where  $M_z$ and $N_z$ are given in Equation \eqref{eq:M-N}.   Moreover,  the following hold
	\begin{enumerate}
		\item[1.]  The maps $z \mapsto N_z$ and $z\mapsto M_z$ are of class $\mathscr C^1$ on $W_{z_0}$ 
		\item[2.] $\Tr\big[\dd M_z M_z^{-1}\big]$ is an exact  one-form on $W_{z_0}$.
	\end{enumerate}
	Let $V_{z_0}\=\set{z\in W_{z_0}}{A_z \ \rm{ invertible}}\subset W_{z_0}$. Then
	\begin{enumerate}
		\item[3.]  $\Tr\big[ \dd A_z A_z^{-1}\big]=\Tr\big[ \dd M_zM_z^{-1}\big]+\Tr \big[\dd N_z N_z^{-1}\big]$ on $V_{z_0}$
		\item[4.] $\Tr\big[\dd N_z  N_z^{-1}\big]$ is a closed but not exact one-form on $V_{z_0}$.
	\end{enumerate}
\end{lem}
\begin{proof}
	We refer the interested reader to \cite[Lemma 2.4 and Proposition 2.5]{PW20} for the proof. 
\end{proof}

\begin{lem} \label{lm:trace_det}
	Given $\varepsilon>0$,	let  $A\in \mathscr C^1\big((-\varepsilon, \varepsilon), \GL(W,H) \big)$ be a path of invertible operators on a Hilbert space $H$ having  the same domain  $W$  and let  $P\in \mathscr C^1\big((-\varepsilon, \varepsilon), \Lin(H) \big)$ be a path of finite rank projections. We assume that
	\begin{itemize}
		\item $ P_tA_t=A_tP_t$ for all $t \in (-\varepsilon, \varepsilon)$
		\item $t \mapsto P_tA_tP_t \in \mathscr C^1((-\epsilon,\epsilon),\Lin(H))$
	\end{itemize}
	Then we get that  $E_t\=P_tA_tP_t|_{\Imm P_t}$ is a linear map from $\Imm P_t$ to $\Imm P_t$ and  by setting   $N_t\=(\Id-P_t)+P_tA_tP_t$, the following equality holds:
	\[
	\Tr[\dd N_t N_t^{-1}]= \Tr[\dd (P_tA_tP_t)P_tA_t^{-1}P_t]=\dd \log\det E_t.
	\]	
\end{lem}
\begin{proof}
	We refer the interested reader to \cite[Proposition 2.6]{PW20} for the proof. 
\end{proof}

\newpage

\vspace{1cm}
\noindent
\textsc{Prof. Alessandro Portaluri}\\
DISAFA\\
Università degli Studi di Torino\\
Largo Paolo Braccini 2 \\
10095 Grugliasco, Torino\\
Italy\\
Website: \url{https://sites.google.com/view/alessandro-portaluri/}\\
E-mail: \email{alessandro.portaluri@unito.it}

\vspace{1cm}
\noindent
\textsc{Prof. Li Wu}\\
Department of Mathematics\\
Shandong University\\
Jinan,Shandong, 250100\\
The People's Republic of China \\
China\\
E-mail: \email{vvvli@sdu.edu.cn}

\vspace{1cm}
\noindent
\textsc{Dr. Ran Yang}\\
School of Science\\
East China  University of Technology\\
Nanchang, Jiangxi, 330013\\
The People's Republic of China \\
China\\
E-mail:\email{yangran2019@ecit.cn}

\end{document}